%% file: main.tex
\documentclass[sn-mathphys-ay]{sn-jnl}

\usepackage[T1]{fontenc}

\usepackage{subcaption}

\usepackage[dvipsnames]{xcolor}
\usepackage{tikz}
\usetikzlibrary{positioning, calc, intersections}

\usepackage{amsmath,amsthm,amsfonts,amssymb,mathtools}
\usepackage{units}
\usepackage{cancel}
\usepackage{resizegather}
\usepackage{ellipsis}

\usepackage{orcidlink}

\newtheoremstyle{theorem}{1.5em}{1.5em}{\itshape}{}{\bf}{.\ \\}{0.5ex}{}
\newtheoremstyle{definition}{1.5em}{1.5em}{}{}{\bf}{.\ \\}{.5ex}{}
\newtheoremstyle{info}{1.5em}{1.5em}{}{}{\bfseries \normalsize }{.\ \\}{0.5ex}{}

\theoremstyle{theorem}
\newtheorem{theorem}{Theorem}

\newtheorem{lemma}{Lemma}
\newtheorem{corollary}{Corollary}

\theoremstyle{definition}
\newtheorem{definition}{Definition}

\theoremstyle{info}
\newtheorem{remark}{Remark}

\newcommand{\PreprintRef}{Tschimmel, D., Saeed, M., Milani, M., Waldherr, S., Hucho, T.: Protein- environment-sensitive computational epitope accessibility analysis from antibody dose-response data. Preprint: arXiv (2024) \url{https://doi.org/10.48550/arXiv.2407.06052} [q-bio.QM]}
\newcommand{\PreprintText}{Preprint: Tschimmel et al. 2024}

\begin{document}

\title{Modelling the effect of antibody depletion for dose-response behavior of common immunostaining protocols}

\author[1]{\fnm{Dominik} \sur{Tschimmel}}\nomail
\author[2]{\fnm{Steffen} \sur{Waldherr}}\nomail
\author[1]{\fnm{Tim} \sur{Hucho}}\nomail

\affil[1]{\orgdiv{Translational Pain Research, Department of Anesthesiology and Intensive Care Medicine}, \orgname{University Hospital Cologne, University of Cologne}, \orgaddress{\postcode{50931} \city{Cologne}, \country{Germany}}}
\affil[2]{\orgdiv{Department of Functional and Evolutionary Ecology}, \orgname{University of Vienna}, \orgaddress{\postcode{1030} \city{Vienna}, \country{Austria}}}

\abstract{
Antibody binding properties for immunostaining applications are often characterized by dose-response curves, which describe the amount of bound antibodies as a function of the antibody concentration applied at the beginning of the experiment. A common model for the dose-response curve is the Langmuir isotherm, which assumes an equilibrium between the binding and unbinding of antibodies. However, for common immunostaining protocols, the equilibrium assumption is violated, and the dose-response behavior is governed by an accumulation of permanently bound antibodies. Assuming a constant antibody concentration, the resulting accumulation model can easily be solved analytically. However, in many experimental setups the overall amount of antibodies is fixed, such that antibody binding reduces the concentration of free antibodies. Solving the corresponding depletion accumulation model is more difficult and seems to be impossible for heterogeneous epitope landscapes. In this paper, we first solve the depletion-free accumulation model analytically for a homogeneous epitope landscape. From the obtained solution, we derive inequalities between the depletion-free accumulation model, the depletion accumulation model, and the Langmuir isotherm. This allows us to characterize the depletion effect for homogeneous epitope landscapes. Next, we generalize the problem to heterogeneous epitope landscapes, where we prove the existence and uniqueness of a solution that behaves as expected from the experimental setting. These natural properties define bounds for the depletion accumulation model. We conclude this paper by applying the bounds to characterize the depletion effect for heterogeneous epitope landscapes. 
}


\maketitle

\input{content.tex}

\backmatter

\bmhead{Author contributions}

The research project was designed by Tim Hucho and Steffen Waldherr. The biological background and biological scope were supervised by Tim Hucho. The modeling was supervised by Steffen Waldherr. Development of the model, deriving the proofs, and writing the manuscript drafts were done by Dominik Tschimmel. The manuscript was revised and approved by all authors.

\begin{appendices}
   \input{appendix.tex}

\end{appendices}

\bibliography{Literature}
\end{document}

%% file: content.tex
\section{Introduction}

The adsorption of molecules from the liquid phase to the surface of a solid is well described in the literature, with many models addressing different binding mechanisms and measurement setups. Among isotherms that describe binding equilibria, the Langmuir isotherm may be the best-known and most widely used model \citep{Alberti_2012}. Consequently, the closely related Langmuir rate equation is often applied to describe binding kinetics.  Interestingly, structurally equivalent but differently named equations are frequently used to describe the binding behavior of antibodies \citep{Latour_2014} \citep{Jarmoskaite_2020} \citep[Section 2.2]{Lauffenburger_1993}.

Let $c(t)$ denote the concentration of free-floating antibodies, $x(t)$ denote the surface concentration of the bound antibody-epitope complexes, and $g$ denote the surface concentration of all epitopes. Then the Langmuir rate equation reads 
\begin{equation}
	\label{eq: Langmuir model}
	x(t) = h_a c(t)(g-x(t))-k_d x(t) \ ,
\end{equation}
with rate constants $h_a$ and $k_d$ determining the binding\slash unbinding rates. The resulting Langmuir isotherm  
\begin{equation}
	\label{eq: Langmuir isotherm}
	x_{\text{eq}} = \frac{g}{1+\frac{k_d}{h_a \cdot c_{\text{eq}}}}\ ,
\end{equation}
describes the equilibrium values $x_{\text{eq}}$ and $c_{\text{eq}}$. For the special case of a constant antibody concentration $c(t) = c$ there is a simple analytical solution of the Langmuir rate equation \eqref{eq: Langmuir model} (cf. \citep{Latour_2014} \citep[Section 2.2]{Lauffenburger_1993}, etc.):
\begin{equation}
	\label{eq: simple Langmuir solution}
	x(t) = x_0 e^{-(h_a\cdot c+k_d)(t-t_0)} + \frac{g}{1+\frac{k_d}{h_a\cdot c}}\left(1- e^{-(h_a\cdot c+k_d)(t-t_0)}\right)\ .
\end{equation}

Despite the widespread use of the Langmuir model, some limitations must be considered. First, the simple analytical solution \eqref{eq: simple Langmuir solution} requires a constant antibody concentration throughout the binding phase. Second, the Langmuir model assumes that all binding sites are identical. In practice, however, the Langmuir model is used even if these conditions are violated. This may be due to negligence, but also because many experiments of interest do not provide the required ideal conditions. As a result, the estimated dissociation constants $K_d \coloneqq \frac{k_d}{h_a}$ can be off by orders of magnitude \citep{Jarmoskaite_2020}.

Because antibody depletion is inevitable in many experiments, \cite{Edwards_1998} investigated correction methods that account for the effects of antibody depletion. Under ideal conditions, the number of free antibodies reduces one-to-one by the number of antibodies that are bound to epitopes. Thus, the antibody concentration is given by $c(t) = c - \beta x(t)$, where $c$ denotes the initial antibody concentration and $\beta$ is the conversion factor between surface concentration and volume concentration.

In addition to correction methods, \cite{Edwards_1998} derived an analytical isotherm and an analytical solution of the Langmuir rate equation \eqref{eq: Langmuir model} for the antibody depletion case. A decade later \cite{Marczewski_2010} provided an alternative expression of the solution by expressing the depletion Langmuir rate equation as a mixed-order equation  (cf. \citep{Liu_2008}). However, the mixed-order solution requires the definition of new variables and constants, which complicates its application and interpretation. This motivated \cite{Salvestrini_2017} to derive a different analytical solution, which turns out to be the solution of \citep{Edwards_1998}.

The second limitation of the Langmuir model, requiring identical binding sites, can easily be addressed by defining different epitope classes $\{(g_i,h_{a,i}, k_{d,i})\}_{i=1}^N$. In this definition, epitopes belong to the same class if they have the same binding properties, i.e. rate constants $h_{a,i}$ and $ k_{d,i}$. When the antibody concentration is constant, the binding processes of the individual epitope classes $x_i(t)$ are independent. Each class behaves according to equation \eqref{eq: simple Langmuir solution}. Since the total surface concentration of bound antibodies is just the sum of the individual epitope classes, it follows that:
\begin{equation}
	\label{eq: discrete Langmuir superposition}
	x(t) = \sum_{i=1}^N x_i(t) = \sum_{i=1}^N x_{0,i} e^{-(h_{a,i}\cdot c+k_{d,i})(t-t_0)} + \frac{g_i}{1+\frac{k_{d,i}}{h_{a,i}\cdot c}}\left(1- e^{-(h_{a,i}\cdot c+k_{d,i})(t-t_0)}\right)\ .
\end{equation}

For heterogeneous systems with many epitope classes, it can become convenient to describe the classes as distribution $g\colon \mathbb{R}^2_{\geq 0}\rightarrow \mathbb{R}, (h_a,k_d)\mapsto g(h_a,k_d)$. The superposition of individual classes then becomes an integral:
 \begin{equation}
	\label{eq: continuous Langmuir superposition}
x(t) = \int_0^\infty \int_0^\infty x_{0} e^{-(h_a\cdot c+k_d)(t-t_0)} + \frac{g(h_a,k_d)}{1+\frac{k_d}{h_a\cdot c}}\left(1- e^{-(h_a\cdot c+k_d)(t-t_0)}\right)\ d h_a d k_d \  .
\end{equation}
This distribution description is particularly useful for inverse problems, where the number of epitope classes is not known in advance. For that reason, \cite{Svitel_2003} applied the distribution description for the analysis of biosensor data. In subsequent applications, mass-transport limitations were investigated by \cite{Svitel_2007} and \cite{Malakhova_2020}. Additional improvements for numerical solutions of the inverse problem were proposed by \cite{Svitel_2008}, \cite{Zhang_2017}, \cite{Forssen_2018}, and \cite{Zhang_2019}. However, note that the general idea to describe heterogeneous binding with integrals is not new and has been used for decades \citep{Langmuir_1918} \citep{Svitel_2003} \citep{House_1978}.

Although the limitations of the Langmuir model can be solved individually, the superposition of epitope classes together for the depletion case cannot easily be solved analytically. Furthermore, the discussed aspects improve upon the Langmuir model, which describes biosensor setups well but does not cover the measurement principle of common immunostaining experiments. In most cases, multiple washing steps are performed during immunocytochemistry\slash immunohistochemistry, before measuring the amount of bound antibodies \citep{Pina_2022}. Thus, any equilibrium that might have settled during the antibody-binding phase gets eventually disrupted, violating the isotherm assumption. Hence, a different model needs to be found to describe the dose-response behavior of immunostaining, that is, the relationship between the antibody concentration and the resulting staining intensity.

To address this mismatch between the model and the experimental protocol, (\PreprintText\footnote{\PreprintRef}) proposed the ``accumulation model'', which combines the pseudo-first-order (PFO) model \citep{Lagergren_1898} with the epitope class approach: 
\begin{equation}
	\label{eq: Accumulation model introduction}
	\frac{d}{dt} x_i(t) = h_{a,i} c(t)(g_i-x_i(t))\ , \qquad \forall\ i\in \{1,\ldots,N\} \ .
\end{equation}
This system of differential equations consists of the Langmuir kinetics without the unbinding term. The key difference is that antibody accumulation is interrupted at a finite time $\tau$ in all immunostaining protocols, before the system reaches its equilibrium point. The characteristic dose-response shape then originates from the concentration-dependent accumulation rate of antibodies, where higher concentrations lead to more antibodies binding in the finite incubation time. Figure \ref{fig: illustration dose-response behavior} illustrates this idea.

\begin{figure}[ht]
	\centering
	\begin{tikzpicture}[xscale = 0.74,yscale = 0.8]
		\filldraw[fill = gray!10, draw = black!50] (-1.2,-1) rectangle (16,6.5);

		\draw[black,thick,-stealth] (0,0) -- (6,0)node[below]{$t$};
		\draw[black,thick,-stealth] (0,0) -- (0,5.5)node[right]{$x(t)$};

		\draw[black!80,dashed,line width = 1pt, name path = t] (4,0)node[below]{$\tau$} --(4,5);

		\draw[red,line width = 1pt, name path = p1] (0,0) to[out=80, in = 180] (6,4.9);
		\draw[red!90!blue,line width = 1pt, name path = p2] (0,0) to[out=75, in = 180] (6,4.8);
		\draw[red!80!blue, line width = 1pt, name path = p3] (0,0) to[out=70, in = 185] (6,4.5);
		\draw[red!60!blue, line width = 1pt, name path = p4] (0,0) to[out=60, in = 190] (6,3.5);
		\draw[red!50!blue, line width = 1pt, name path = p5] (0,0) to[out=40, in = 187] (6,2.1);
		\draw[red!40!blue, line width = 1pt, name path = p6] (0,0) to[out=20, in = 185] (6,1.2);
		\draw[red!20!blue, line width = 1pt, name path = p7] (0,0) to (6,0.8);
		\draw[red!10!blue, line width = 1pt, name path = p8] (0,0) -- (6,0.5);
		\draw[blue, line width = 1pt, name path = p9] (0,0) -- (6,0.3);

		\path [name intersections={of=p1 and t,by=t1}];
		\path [name intersections={of=p2 and t,by=t2}];
		\path [name intersections={of=p3 and t,by=t3}];
		\path [name intersections={of=p4 and t,by=t4}];
		\path [name intersections={of=p5 and t,by=t5}];
		\path [name intersections={of=p6 and t,by=t6}];
		\path [name intersections={of=p7 and t,by=t7}];
		\path [name intersections={of=p8 and t,by=t8}];
		\path [name intersections={of=p9 and t,by=t9}];

		\fill[red] (t1) circle (0.1);
		\fill[red!90!blue] (t2) circle (0.1);
		\fill[red!80!blue] (t3) circle (0.1);
		\fill[red!60!blue] (t4) circle (0.1);
		\fill[red!50!blue] (t5) circle (0.1);
		\fill[red!40!blue] (t6) circle (0.1);
		\fill[red!20!blue] (t7) circle (0.1);
		\fill[red!10!blue] (t8) circle (0.1);
		\fill[blue] (t9) circle (0.1);

		\draw[black,thick,-stealth] (8,0) -- (15.5,0)node[below left]{$\log(c)$};
		\draw[black,thick,-stealth] (8,0) -- (8,5.5)node[right]{$x(\tau)$};

		\coordinate (T1) at ($(t1) + (10.9,0)$);
		\coordinate (T2) at ($(t2) + (10.1,0)$);
		\coordinate (T3) at ($(t3) + (9.3,0)$);
		\coordinate (T4) at ($(t4) + (8.5,0)$);
		\coordinate (T5) at ($(t5) + (7.7,0)$);
		\coordinate (T6) at ($(t6) + (6.9,0)$);
		\coordinate (T7) at ($(t7) + (6.1,0)$);
		\coordinate (T8) at ($(t8) + (5.3,0)$);
		\coordinate (T9) at ($(t9) + (4.5,0)$);

		\draw[line width = 1pt,smooth] plot coordinates {(T1) (T2) (T3) (T4) (T5) (T6) (T7) (T8) (T9)};

		\draw[dashed, black!70] (t1) -- (T1);
		\fill[red] (T1) circle (0.1);

		\draw[dashed, black!70] (t2) -- (T2);
		\fill[red!90!blue] (T2) circle (0.1);

		\draw[dashed, black!70] (t3) -- (T3);
		\fill[red!80!blue] (T3) circle (0.1);

		\draw[dashed, black!70] (t4) -- (T4);
		\fill[red!60!blue] (T4) circle (0.1);

		\draw[dashed, black!70] (t5) -- (T5);
		\fill[red!50!blue] (T5) circle (0.1);

		\draw[dashed, black!70] (t6) -- (T6);
		\fill[red!40!blue] (T6) circle (0.1);

		\draw[dashed, black!70] (t7) -- (T7);
		\fill[red!20!blue] (T7) circle (0.1);

		\draw[dashed, black!70] (t8) -- (T8);
		\fill[red!10!blue] (T8) circle (0.1);

		\draw[dashed, black!70] (t9) -- (T9);
		\fill[blue] (T9) circle (0.1);

	\end{tikzpicture}
	\caption{Illustration, reproduced from (\PreprintText), to depict the relationship between binding rates and dose-response behavior. The left-hand plot illustrates binding curves with different antibody-concentration-dependent binding rates (red = high concentration, blue = low concentration). The right-hand side plot shows the resulting dose-response behavior when the antibody incubation phase is stopped for each concentration after $t=\tau$.}
	\label{fig: illustration dose-response behavior}
\end{figure}

When the antibody concentration is constant, the accumulation model can easily be solved and the solution remains valid even for arbitrary units. (\PreprintText) used this property to develop computational analysis tools for immunofluorescence microscopy. However, in many experimental settings, the concentration of free antibodies decreases when antibodies bind. Despite this mismatch of assumption and reality, the resulting model provided consistent applications, suggesting that the antibody depletion effect does not alter the model properties too much.

Without a mathematical analysis of the antibody depletion case, it remains questionable in which cases the assumption of a constant antibody concentration can be used for systems with antibody depletion.
Furthermore, the data correction method to undo the depletion effect as described in (\PreprintText), while well motivated, has not yet been connected to the model that actually describes systems with depletion. It was a simple worst-case correction method. Finally, there may be better approximations of the accumulation model with antibody depletion than just ignoring the depletion and applying a worst-case correction, especially for experiments where the (surface) concentrations can be precisely measured.

To address the aforementioned issues, we focus on the mathematical description of the antibody accumulation process in this paper. First, we describe the experimental setting and select convenient units. Next, we solve the depletion accumulation model analytically for a single epitope class and derive inequalities between the depletion accumulation model, the depletion-free accumulation model, and the Langmuir isotherm. With these inequalities, we provide heuristics for the antibody depletion behavior. 

Since an analytical solution seems impossible for multiple epitope classes, we prove the existence and uniqueness of a solution that satisfies the experimentally expected behavior. From these properties, we derive bounds for the solution to estimate its behavior. Finally, we prove the antibody depletion correction proposed in (\PreprintText) and derive an additional minimal depletion correction.

\section{The experimental setting and definitions}
\label{sec: experimental setting}

In this paper, we consider cells that are attached to a solid surface in a container (e.g. a well in a well plate or a Petri dish) and antibodies that are dissolved homogeneously in a liquid (i.e. the antibody solution). Furthermore, we assume idealized conditions. For example, external parameters such as temperature and pressure are constant. For the binding of antibodies, we distinguish between two cases:
\begin{description}
	\item[\textbf{depletion:}]\ \\ The antibody solution is added to the container at time $t=0$. Until $t = \tau$, the antibodies are incubated. During this incubation phase, no new antibodies are added to the system.

	\item[\textbf{depletion-free:}]\ \\ The container is connected to a reservoir and an outlet. At time $t = 0$, the antibody solution begins to flow through the container until $t=\tau$.
\end{description}
In either case, we assume that multiple washing steps follow the incubation phase before the amount of bound antibodies is measured. 

\begin{remark}[\textbf{antibody concentration}]
Whenever we speak about the antibody concentration (of the antibody solution) we always consider the concentration of unbound antibodies. Thus, the binding of antibodies in the depletion case reduces the antibody concentration.
\end{remark}

In this setting, the quantities of interest are best expressed as concentrations\slash surface concentrations:
\begin{equation*}
	[c(t)] = \unit{L^{-3}}, \quad [g] = [x(t)] = \unit{L^{-2}}\quad \Rightarrow \quad [\tfrac{d}{dt} x(t)] = \unit{L^{-2}s^{-1}} \quad \Rightarrow\quad [h_a] = \unit{L^3s^{-1}}\ .
\end{equation*}
Yet, it will be convenient for calculations to convert the units of the quantities. This can already be seen for a single epitope class. 

In the depletion case, the antibody concentration is given by
\[c(t) = c - \beta x(t)\ ,\]
where $\beta$ converts the surface concentration $x(t)$ of bound antibodies into an equivalent volume concentration. For example, if the cells are attached to a surface with surface area $S$ and the antibody solution occupies the volume $V$, the conversion factor is $\beta = \frac{S}{V}$. With this expression for the antibody concentration, the accumulation model for a single epitope class becomes
\begin{equation}
	\label{eq: association term natural units}
	\frac{d}{dt} x(t) = h_a c(t) (g-x(t)) = h_a \left(c -\beta x(t)\right)(g-x(t))\ .
\end{equation}
We can simplify the right-hand side by absorbing the conversion factor into the unit of the antibody concentration
\[\begin{aligned}
	h_a \left(c -\beta x(t)\right)(g-x(t)) &= h_a  \frac{\beta}{\beta} \left(c -\frac{S}{V}x(t)\right)(g-x(t))\\
	&= \beta h_a \left( \frac{1}{\beta}c -x(t)\right)(g-x(t))\\
	& \eqqcolon k_a(a-x(t))(g-x(t))\ .
\end{aligned}\]
In the last step, we defined new units such that
\[a \coloneqq \frac{1}{\beta} c\qquad \text{and}\qquad k_a \coloneqq \beta h_a\ ,  \]
which expresses the antibody concentration and the rate constant in terms of surface concentrations. Changing the units in this way does not affect the structure of the depletion-free case, since
\[h_a c (g-x(t)) = h_a \frac{\beta}{\beta}c(g-x(t)) = k_a a (g-x(t))\ .\]

\begin{remark}[\textbf{Units and usage of the word ``concentration''}]
	\label{rem: unit conversion}
	In the rest of the paper, we will only use the converted antibody concentration $a$, $a(t)$, etc., and the converted rate constant $k_a$. For brevity, we will speak only about antibody concentration and rate constant, despite the unit conversion. In the same way, we will drop the word ``surface'' for the surface concentrations $g$ and $x$ and call them ``concentrations'' as well.
\end{remark}

We may conclude this section by defining the equations of interest for this paper. Since neither non-existing epitopes ($g=0$) nor epitopes that cannot be bound by antibodies ($k_a = 0$) constitute proper epitopes, we define

\begin{definition}
	An \textbf{epitope class} $(g,k_a)$ is a collection of epitopes that have the same binding rate constant $k_a > 0$. The empty class without epitopes, that is, with concentration $g = 0$, does not constitute a valid epitope class.
\end{definition}

We use the technical name ``class'' to highlight that this distinction of epitopes is only based on the binding behavior of epitopes (for a given antibody). Biologically different epitopes can belong to the same epitope class. In the same way, epitopes with identical biological structures may end up in different classes if they appear at different locations in the cell.

Next, let us fix names for the equations of interest to enable convenient reference.

\begin{definition}
    \label{def: accumulation IVPs}
	Let $\{(g_i,k_{a,i})\}_{i=1}^N$ be epitope classes and let $a>0$ be the initial antibody concentration. Then we call the initial value problem (IVP)
	\[\frac{d}{dt} x_i(t) = k_a \left(a-\sum_{j=1}^N x_j(t)\right) (g_i-x_i(t))\ , \quad x_i(0) = 0 \qquad \forall\ i \in \{1,\ldots,N\} \]
	\textbf{depletion accumulation IVP} and 
	\[\frac{d}{dt} x_{\textnormal{DF},i}(t) = k_a a (g_i-x_{\textnormal{DF},i}(t))\ , \quad x_{\textnormal{DF},i}(0) = 0 \qquad \forall\ i \in \{1,\ldots,N\} \]
	\textbf{depletion-free accumulation IVP}.
\end{definition}

The accumulation IVPs are autonomous systems:
\[\frac{d}{dt} \boldsymbol{x}(t) = f(\boldsymbol{x}(t))\ , \qquad f\colon \mathbb{R}^N\rightarrow \mathbb{R}^N\ ,\  \boldsymbol{x}\mapsto \begin{pmatrix}
    k_{a,1} \left(a-\sum_{j=1}^N x_j\right) (g_1-x_1) \\
    \vdots\\
    k_{a,N} \left(a-\sum_{j=1}^N x_j\right) (g_N-x_N) \\
\end{pmatrix}  \]
and 
\[\frac{d}{dt} \boldsymbol{x}_{\textnormal{DF}}(t) = h(\boldsymbol{x}_{\textnormal{DF}}(t))\ , \qquad h\colon \mathbb{R}^N\rightarrow \mathbb{R}^N\ , \ \boldsymbol{x}\mapsto \begin{pmatrix}
    k_{a,1} a (g_1-x_1) \\
    \vdots\\
    k_{a,N} a (g_N - x_N)
\end{pmatrix} \ , \]
where, both $f$ and $g$ are locally Lipschitz continuous on $\mathbb{R}^N$.

Since the measurement of bound antibodies is performed after the incubation phase $[0,\tau]$, we are particularly interested in the values $x_i(\tau)$ and $x_{\textnormal{DF}, i}(\tau)$. So far, our notation does not indicate the system parameters, i.e. the initial antibody concentration $a$ and the epitope classes $\{(g_i,k_{a,i})\}_{i=1}^N$. This will become problematic for the description of dose-response curves, where we have to consider different antibody concentrations. Thus, let us add the system parameters as additional arguments, whenever necessary for disambiguation. We may then write, for example,
\[x_i(t\ ; a, \{(g,k_{a,i})\}_{i=1}^N)\ , \qquad x_i(t\ ; a)\qquad \text{or}\qquad x_i(t)\ ,\]
depending on the context.

\begin{definition}
	\label{def: accumulation models}
	Let $\{(g_i,k_{a,i})\}_{i=1}^N$ be epitope classes and let $\tau \in [0,T)$. If the depletion accumulation IVP has a unique solution $\boldsymbol{x}(t\ ; a ,  \{(g_i,k_{a,i})\}_{i=1}^N)$ for $t\in [0,T)$ and $a \in I\subseteq \mathbb{R}_{> 0}$, we call $\boldsymbol{x}(\tau\ ; a ,  \{(g_i,k_{a,i})\}_{i=1}^N)$ the \textbf{depletion accumulation model} (over $I$). In the same way, we define the \textbf{depletion-free accumulation model} $\boldsymbol{x}_{\textnormal{DF}}(\tau\ ; a, \{(g_i,k_{a,i})\}_{i=1}^N)$.
\end{definition}

In general, immunofluorescence microscopy cannot distinguish between the epitopes to which the antibodies have bound. Only differently labeled antibodies can be distinguished. Instead, the overall fluorescence intensity
\[\sum_{j=1}^{N} x_j(\tau\ ; a, \{(g_i,k_{a,i})\}_{i=1}^N)\qquad \text{or}\qquad \sum_{j=1}^N x_{\textnormal{DF},i}(\tau\ ; a , \{(g_i, k_{a,i})\}_{i=1}^N \]
is measured. So, we define

\begin{definition}
	Let $\boldsymbol{x}(\tau\ ; a , \{(g_i,k_{a,i})\}_{i=1}^N)$ be a depletion accumulation model, then we call
	\[X(a\ ; \{(g_i,k_{a,i})\}_{i=1}^N) \coloneqq \sum_{j=1}^{N} x_j(\tau\ ; a , \{(g_i,k_{a,i})\}_{i=1}^N)\]
	\textbf{depletion dose-response curve}. In the same way, we define the \textbf{depletion-free dose-response curve} $X_{\textnormal{DF}}(a\ ; \{(g_i,k_{a,i})\}_{i=1}^N)$.
\end{definition}

\section{Accumulation model for a single epitope class}

As a first step for the analysis of antibody depletion, we should consider the simple case of a single epitope class, which permits analytical solutions of the underlying differential equations. From these solutions, the properties of the accumulation model can easily be derived with methods from basic calculus, setting the expectations for properties of the general case with arbitrarily many epitope classes.

\subsection{Solving the accumulation IVPs}
Observe, that we have already encountered a solution of the depletion-free accumulation IVP in the introduction:
\begin{equation}
	\label{eq: depletion-free IVP solution}
	x_{\textnormal{DF}}(t) = g(1-e^{-k_a a t})\ .
\end{equation}
Solving the full accumulation IVP requires a little more effort but still consists of basic calculus. 

\begin{lemma}
	\label{lemma: depletion accumulation IVP single epitope class}
	Let $(g,k_{a})$ be the only epitope class. When $a\neq g$,
	\[x(t\ ; a) = \frac{a g\left(1-e^{k_a(a -g)t}\right)}{g-a e^{k_a(a-g)t}}\]
	solves the depletion accumulation IVP for $t \geq 0$. When $a = g$,
\[x(t\ ;g) = \frac{g}{1+\frac{1}{k_a g t}} \]
	solves the depletion accumulation IVP for $t \geq 0$.
\end{lemma}
\begin{proof}
	First, we observe that the expressions for $x(t\ ; a)$ and $x(t\ ; g)$ are well-defined and finite for all $t\geq 0$.  Next, we calculate the derivatives
	\[\frac{d}{dt} x(t\ ; a) = \frac{k_a a g (a-g)^2 e^{k_a(a-g)t}}{(g-a e^{k_a(a-g)t})^2}  \qquad \text{and}\qquad  \frac{d}{dt} x(t\ ; g) = \frac{k_a g^2}{(1+k_a g t)^2}\ ,\]
	which are again well-defined and finite for all $t\geq 0$. Plugging in the expressions for $x(t\ ; a)$ and $x(t\ ; g)$ into the respective differential equations, we obtain 
	\[k_a (a- x(t\ ; a))(g-x(t\ ; a)) = \frac{k_a a g (a-g)^2 e^{k_a(a-g)t}}{(g-a e^{k_a(a-g)t})^2} \]
	and 
	\[k_a (g- x(t\ ; g))(g-x(t\ ; g)) = \frac{k_a g^2}{(1+k_a g t)^2}\ , \]
	which shows that $x(t\ ; a)$ and $x(t\ ; g)$ solve the respective differential equations. Finally, we calculate that $x(0\ ; a) = 0$ and that 
	\[x(0\ ;g)\coloneqq \lim_{t\searrow 0} \frac{g}{1+\frac{1}{k_a g t}} = 0\ .\]
\end{proof}

In summary, we have proven that

\begin{theorem}
	\label{thm: unique solution single epitope class}
	Let $(g,k_a)$ be the only epitope class.  Then the accumulation IVPs have the following unique solutions for $t \geq 0$:
	\[x_{\textnormal{DF}}(t) = g(1-e^{-k_a a t}) \qquad \text{and}\qquad  x(t) = \left\{\begin{array}{ll}
		 \frac{a g\left(1-e^{k_a(a -g)t}\right)}{g-a e^{k_a(a-g)t}} &\ , \ a \neq g\\[1em]
		 \frac{g}{1+\frac{1}{k_a g t}} & \ ,\ a = g
	\end{array} \right.  \]
\end{theorem}
\begin{proof}
	That the expressions for $x_{\textnormal{DF}}(t)$ and $x(t)$ solve the accumulation IVPs is the statement of Lemma \ref{lemma: depletion accumulation IVP single epitope class}. 

	For the uniqueness, we recall that both accumulation IVPs are autonomous systems with locally Lipschitz continuous functions. Thus, the general existence and uniqueness theorem for initial value problems \cite[Chapter II, \S 6, Theorem VII]{Walter_1998} applies. Since we have found solutions that exist for all $t\geq 0$, the uniqueness holds for $t\geq 0$.  This can easily be seen with the same steps that are used in the proof of \cite[Chapter II, \S 6, Theorem VII]{Walter_1998}.
 
\end{proof}

Next, we observe that the depletion accumulation IVP remains unchanged when we exchange the values for $a$ and $g$
\[k_a (a-x(t))(g-x(t)) = k_a(g-x(t))(a-x(t))\ .\]
Yet, this means that also the solution of the depletion accumulation model should remain unchanged when the values for $a$ and $g$ are exchanged.
\begin{lemma}
	\label{lemma: exchanging a and g}
	Let $(g,k_a)$ be the only epitope class and let $x$ be the solution of the depletion accumulation IVP, then it holds that
	\[x(t\ ; a , (g,k_a)) = x(t\ ; g, (a,k_a))\ .\]
\end{lemma}
\begin{proof}
	For $a = g$ this is obvious. For $a\neq g$ we calculate
	\begin{align*}
			x(t\ ;a,(g,k_a)) &= \frac{ag(e^{k_a g t} - e^{k_a  a t})}{ge^{k_a  g t} - ae^{k_a  a t}} = \frac{-ga(e^{k_a a t} - e^{k_a  g t})}{- (ae^{k_a  a t}  - ge^{k_a  g t})} = \frac{ga(e^{k_a a t} - e^{k_a  g t})}{(ae^{k_a  a t}  - ge^{k_a  g t})}\\[1em]
	&= x(t\ ; g, (a,k_a))\ .
	\end{align*}
\end{proof}

With this symmetry, we can show that the solutions of the accumulation IVPs satisfy the physically expected bounds. That is, there are no negative surface concentrations, nor can the amount of bound antibodies exceed the amount of epitopes (or the available amount of antibodies for the depletion case).

\begin{theorem}
	\label{thm: bounded solutions single epitope}
	Let $(g,k_a)$ be the only epitope class and let $x$ and $x_{\textnormal{DF}}$ be the solutions of the respective accumulation IVPs. Then it holds for all $t\in [0,\infty]$ that
	\[0 \leq  x_{\textnormal{DF}}(t\ ; a) \leq g \qquad \text{and}\qquad 0 \leq x(t\ ; a)\leq \min\{a,g\}\ .\]
\end{theorem}
\begin{proof}
	The statements for $x_{\textnormal{DF}}$ and for $x(t\ ; g, (g,k_a))$ are obvious. When $a<g$, the numerator and denominator of the expression for $x(t\ ; a, (g,k_a))$ are both positive (or zero), which shows that $0 \leq x(t\ ; a, (g,k_a))$. Furthermore,
	\[\frac{g}{x(t\ ; a , (g,k_a))} = \frac{g- ae^{k_a(a-g)t}}{a(1-e^{k_a (a-g)t})} \leq  \frac{a- ae^{k_a(a-g)t}}{a(1-e^{k_a (a-g)t})} = 1\ , \]
	\[\Rightarrow\qquad x(t\ ; a, (g,k_a)) \leq g\ .\]

	When $a>g$, the numerator and denominator of the expression for $x(t\ ; a, (g,k_a))$ are both negative (or zero), which again shows that $0 \leq x(t\ ; a, (g,k_a))$. Furthermore, 
	\[\frac{g}{x(t\ ; a , (g,k_a))} = \frac{g- ae^{k_a(a-g)t}}{a(1-e^{k_a (a-g)t})} \leq  \frac{g- ge^{k_a(a-g)t}}{a(1-e^{k_a (a-g)t})} = \frac{g}{a} \leq 1\ , \]
	\[\Rightarrow\qquad x(t\ ; a, (g,k_a)) \leq g\ .\]

	Since $x(t\ ; a, (g,k_a)) = x(t\ ; g, (a,k_a))$ by Lemma \ref{lemma: exchanging a and g}, it also follows that $x(t\ ; a , (g,k_a))\leq a$.
\end{proof}

\subsection{Properties of the accumulation models and comparison to the Langmuir isotherm}

With the closed-form expressions for the unique solutions of the single-epitope-class accumulation IVPs (Theorem \ref{thm: unique solution single epitope class}) we have also obtained closed-form expressions for the corresponding accumulation models (cf. Definition \ref{def: accumulation models}).

\begin{theorem}
	\label{thm: accumulation model single epitope class}
	Let $(g,k_a)$ be the only epitope class. Then the accumulation models are
		\[
		x_{\textnormal{DF}}(\tau\ ; a) = g(1-e^{-k_a a \tau}) \quad \text{and}\quad  x(\tau\ ; a) = \left\{\begin{array}{ll}
		 \frac{a g\left(1-e^{k_a(a -g)\tau}\right)}{g-a e^{k_a(a-g)\tau}} &\ , \ a \neq g\\[1em]
		 \frac{g}{1+\frac{1}{k_a g \tau}} & \ ,\ a = g
	\end{array} \right.
	\]
	for all $a > 0$ and all $\tau \geq 0$. Both models are continuously differentiable as a function of $a$ with derivatives
	\[ \frac{d}{da}x_{\textnormal{DF}}(\tau\ ; a) = k_a g \tau e^{-k_a a \tau}\] 
	and
	\[\frac{d}{da} x(\tau \ ; a) =\left\{\begin{array}{ll}
		\frac{g(g-e^{k_a(a-g)\tau}(k_a a \tau (g-a) +g))}{(g-ae^{k_a(a-g)\tau})^2} &\ ,\ a \neq g\\[1em]
		\frac{k_a g \tau (2+k_a g \tau)}{2(1+k_a g \tau)^2} & \ , \  a = g
	\end{array}\right. \ . \]
\end{theorem}
\begin{proof}
The existence and the expressions for $x_{\textnormal{DF}}(\tau\ ; a)$ and $x(\tau\ ; a)$ follow from Theorem \ref{thm: unique solution single epitope class} and Definition \ref{def: accumulation models}.

Let now $\tau > 0$.  That $x_{\textnormal{DF}}(\tau \ ; a)$ is continuously differentiable and the expression for the derivative $\frac{d}{da} x_{\textnormal{DF}}(\tau \ ; a)$ can easily be calculated. Furthermore, the derivative $\frac{d}{da} x(\tau\ ; a)$ can easily be calculated at $a\neq g$. The derivative $\frac{d}{da} x(\tau\ ; a)$ at $a=g$ needs to be calculated from the definition of derivatives, which yields
\[\left.\frac{d}{d a } x(a) \right|_{a = g}\  \overset{\text{def}}{=}\  \lim_{a \rightarrow g} \frac{x(\tau\ ; a)-x(\tau\ ; g)}{a-g} = \frac{k_a g \tau (2+k_a g \tau)}{2(1+k_a g \tau)^2}\ ,\]
after applying L'H\^{o}pital's rule twice.
That $\frac{d}{da} x(\tau\ ; a)$ is continuous in $a= g$, i.e.  $\lim_{a \rightarrow g} \left. \frac{d}{d a} x(a)\right|_{a \neq g} = \left.\frac{d}{da} x(a)\right|_{a = g}$, can also be shown by applying L'H\^{o}pital's rule twice.  Furthermore, since all denominators that appear in $\frac{d}{da} x(\tau\ ; a)$ are non-zero for their respective cases, $\frac{d}{da} x(\tau\ ; a)$ is continuous in $a$ as a rational function.

For $\tau = 0$ it holds that $x_{\textnormal{DF}}(0\ ; a) = x(0\ ; a) = 0$, which implies that the derivatives with respect to $a$ must be zero. It can easily be checked that plugging in $\tau = 0$ into the expressions for the derivatives yields zero.
\end{proof}

\begin{remark}
	For a single epitope class, the distinction between the accumulation model and the dose-response curve is pointless, as $x_{\textnormal{DF}}(\tau\ ; a) = X_{\textnormal{DF}}(a)$ and $x(\tau\ ; a) = X(a)$. Thus, we treat ``accumulation model'' and ``dose-response curve'' as synonyms in this section, and stick to the notation $x_{\textnormal{DF}}(a)$, $x(a)$, etc.
\end{remark}

With the derivatives of the accumulation models as a function of the initial antibody concentration, we can show the natural dose-response property. Increasing the initial antibody concentration must not decrease the amount of bound antibodies.

\begin{theorem}
    Let $(g,k_a)$ be the only epitope class and let $\tau > 0$, then both accumulation models are monotonically increasing functions of the antibody concentration:
    \[\frac{d}{da}x_{\textnormal{DF}}(\tau\ ; a) \geq 0 \qquad \text{and}\qquad \frac{d}{da} x(\tau\ ; a) \geq 0 \qquad \forall\ a > 0\ .\]
\end{theorem}
\begin{proof}
    The statement is obvious for $x_{\textnormal{DF}}(\tau\ ; a)$ and for the $a=g$ case of $x(\tau\ ; a)$. For $x(\tau\ ; a)$ and $a\neq g$, the denominator of $\frac{d}{da}x(\tau\ ; a)$ is positive. Thus, the sign is determined only by the numerator. We observe that this numerator is positive if
    \[ d \coloneqq g(1-e^{k_a (a-g)\tau}) + k_a a (a-g)\tau e^{k_a (a-g)\tau} > 0\ .  \] 
    For $a<g$, the first summand $g(1-e^{k_a (a-g)\tau})$ is positive, while the second summand $k_a a (a-g)\tau e^{k_a (a-g)\tau}$ is negative. Thus, the condition $d > 0$ is equivalent to
    \[\frac{-k_a a (a-g)\tau e^{k_a (a-g)\tau}}{g(1-e^{k_a (a-g)\tau })} < 1\ .\]
    To check that this is the case, we define $z = -k_a(a-g)\tau$ and calculate
    \begin{align*}
       \frac{-k_a a (a-g)\tau e^{k_a (a-g)\tau}}{g(1-e^{k_a (a-g)\tau })} &\leq  \frac{-k_a a (a-g)\tau e^{k_a (a-g)\tau}}{a(1-e^{k_a (a-g)\tau })} = \frac{-k_a  (a-g)\tau e^{k_a (a-g)\tau}}{1-e^{k_a (a-g)\tau }}\\[1em]
       &\quad = \frac{ze^{-z}}{1-e^{-z}} = \frac{z}{e^{z} - 1}  < 1
    \end{align*} 
    for all $z > 0$. This shows that $d>0$ for $a< g$. 

    For $a> g$, the first summand $g(1-e^{k_a (a-g)\tau})$ is negative, while the second summand $k_a a (a-g)\tau e^{k_a (a-g)\tau}$ is positive. Here, the condition $d > 0$ is equivalent to
    \[\frac{-g(1-e^{k_a (a-g)\tau })}{k_a a (a-g)\tau e^{k_a (a-g)\tau}} < 1\ .\]
    To check that this is the case, we now define $z = k_a(a-g)\tau$ and calculate
   \begin{align*}
       \frac{-g(1-e^{k_a (a-g)\tau })}{k_a a (a-g)\tau e^{k_a (a-g)\tau}} &\leq \frac{-a(1-e^{k_a (a-g)\tau })}{k_a a (a-g)\tau e^{k_a (a-g)\tau}} = \frac{e^{k_a (a-g)\tau } - 1}{k_a (a-g)\tau e^{k_a (a-g)\tau}}\\[1em]
       &\quad = \frac{e^z -1}{ze^z} = \frac{1-e^{-z}}{z} < 1
   \end{align*} 
    for all $z> 0$. This shows that $d>0$ for $a > g$.
\end{proof}

Before we continue to analyze the properties of the accumulation models, it is worthwhile to include the Langmuir isotherm
\[x_{\text{L}}(a) = \frac{g}{1+\frac{k_d}{k_a a}} \]
in the analysis, as its dose-response behavior is well known.
However, the Langmuir isotherm describes a dynamic equilibrium, determined by the binding rate constant $k_a$ and the unbinding rate constant $k_d$. On the other hand, the accumulation model depends on the binding rate constant $k_a$ and the incubation time $\tau$. This complicates a direct comparison. 

However, the dose-response behavior of the Langmuir isotherm does not depend on the individual values of $k_a$ and $k_d$, but on the quotient $K_a = \frac{k_a}{k_d}$, which is called the binding constant. This binding constant determines the dose-response behavior and thus characterizes the system.
In the same way, we can define a system-characterizing constant for the accumulation model. We observe that $k_a$ and $\tau$ always appear as a product $k_a \tau$, and we define
\[K\coloneqq k_a \tau\ , \]
similar to $K_\tau = \frac{1}{k_a \tau}$ from (\PreprintText).  With this constant, the models read
\[x(a) = \left\{\begin{array}{ll}
		\frac{a g\left(1-e^{K (a -g)}\right)}{g-a e^{K (a-g)}} &\ , \ a \neq g\\[1em]
		\frac{g}{1+\frac{1}{K g}} & \ ,\ a = g
\end{array} \right.  \qquad  \begin{array}{l}
	x_{\text{DF}}(a) = g(1-e^{-K a })\\[1em]
	x_{\text{L}}(a) = \frac{g}{1+\frac{1}{K a}}
\end{array}\quad .\]
Despite describing different processes, we can now compare the dose-response behavior of the Langmuir isotherm and the accumulation model, by assuming identical values for the system-characterizing constants $K_a = K$. 

\begin{theorem}
	\label{thm: SEC model inequalities}
	For $K , a ,g  \in (0,\infty)$ it holds that
	\begin{align*}
	x(a) < x_{\textnormal{L}}(a) < x_{\textnormal{DF}}(a) & \qquad \text{if} \ a < g\\[1em]
	x(a) = x_{\textnormal{L}}(a) < x_{\textnormal{DF}}(a) & \qquad \text{if} \ a = g\\[1em]
	 x_{\textnormal{L}}(a) < x(a)  < x_{\textnormal{DF}}(a) & \qquad \text{if} \ a > g\ ,
	\end{align*}
    where $x(a) = x_{\textnormal{L}}(a)$ if $a = g$.
\end{theorem}

\begin{proof}
    The case $x(g) = x_{\textnormal{L}}(g)$ is obvious. Thus, let $a < g$, then
    \[\frac{x(a)}{x_{\textnormal{L}}(a)} = \frac{a g (1+K a)(1-e^{K (a-g)})}{K a g (g- a e^{K (a-g)})} = \frac{1+K a - e^{K(a-g)} - Ka e^{K(a-g)}}{K g - K a e^{K(a-g)}} \eqqcolon \frac{N}{D} \ , \]
    where we defined the numerator term $N$ and the denominator term $D$ in the last step. Note that $1+K a$ is always positive. Since $K(a -g) < 0 $ for $a < g$, the term $1-e^{K(a-g)}$ is also positive. Thus, the numerator $N$ is positive. Furthermore, $g-ae^{K(a-g)}$ is also positive, such that the denominator is positive. Thus, $\frac{N}{D}< 1$ if $N-D < 0$:
    \[N-D = 1 + K(a-g) -e^{K(a-g)} \eqqcolon \Delta(K(a-g))\ ,\]
    where we define the function $\Delta(z) = 1+z-e^{z}$. It can easily be shown that $\Delta(z) < 0$ for all $z \neq 0$. Hence, we have proven that
    \[a < g \qquad \Rightarrow \qquad \frac{x(a)}{x_{\textnormal{L}}(a)} = \frac{N}{D} < 1 \qquad \Rightarrow \qquad x(a) < x_{\textnormal{L}}(a)\ .\]
    Let now $a > g$, then $K(a-g) > 0$ and thus $1-e^{K(a-g)} < 0 $. It follows that $N<0$. In the same way, since it follows that $g-ae^{K(a-g)}<0$, i.e. $D < 0$. Thus, $\frac{N}{D} > 1$ if $(-N)-(-D) = D-N > 0$. We observe that
    \[D-N = -1 - K(a-g) + e^{K(a-g)} = -\Delta(K(a -g)) > 0\ ,\]
    since $\Delta(z)< 0$ for $z\neq0$. Thus, we have proven
    \[a > g \qquad \Rightarrow \qquad \frac{x(a)}{x_{\textnormal{L}}(a)} = \frac{N}{D} \geq 1 \qquad \Rightarrow \qquad x_{\textnormal{L}}(a) < x(a)\ .\]
    
    Next, we consider $x_{\textnormal{DF}}(a)-x_{\textnormal{L}}(a)$ for $a \leq g$:
    \[x_{\textnormal{DF}}(a)-x_{\textnormal{L}}(a) = g(1-e^{-K a}) - \frac{K a g}{1+ K a } =\frac{g(1+K a)(1-e^{-Ka})- K a g}{1+K a} \eqqcolon \frac{N_2}{D_2}\ . \]
    Here, we defined a new numerator term $N_2$ and a new denominator term $D_2$. Since $1+K a$ is positive, i.e. $D_2>0$, the sign of $x_{\textnormal{DF}}(a)-x_{\textnormal{L}}(a)$ is determined by $N_2$. We calculate
    \[N_2 = g(1-e^{-Ka} - K a e^{-Ka}) \eqqcolon g\cdot \Delta_2(K a)\ , \]
    where we defined $\Delta_2(z) = 1-e^{-z}- z e^{-z}$. It can easily be shown that $\Delta_2(z) > 0 $ for all $z > 0$. Thus, it follows that $N_2 > 0$, which leads to 
    \[a\leq g \qquad \Rightarrow \qquad x(a) \leq  x_{\textnormal{L}}(a) < x_{\textnormal{DF}}(a)\ , \]
    where we used the inequality $x(a)\leq x_{\textnormal{L}}(a)$ for $a \leq g$, that was proven before.
    
    Finally, we consider $\frac{x_{\textnormal{DF}}(a)}{x(a)}$ for $a \geq g$. For this, we use the symmetry of $x(a)$ with respect to $a$ and $g$ from Lemma \ref{lemma: exchanging a and g} to calculate
    \[\frac{x_{\textnormal{DF}}(a)}{x(a)} =\frac{(a -g e^{K(g-a)})(1-e^{-K a})}{a (1-e^{K(g-a)})}\ . \]
    As before, it can be seen that both the denominator and the numerator are non-negative. By replacing the first $g$  with an $a$ we obtain the following inequality (note that $a \geq g$)
    \[\frac{x_{\textnormal{DF}}(a)}{x(a)}  \leq \frac{(a -a e^{K(g-a)})(1-e^{-K a})}{a (1-e^{K(g-a)})} = 1-e^{-K a} < 1\ . \]
    Thus, it holds that
    \[a \geq g\qquad \Rightarrow \qquad x_{\textnormal{L}}(a) \leq  x(a) < x_{\textnormal{DF}}(a)\ ,\]
    where we used the inequality $x_{\textnormal{L}}(a) \leq x(a)$ for $a\geq 0$, that was proven before.
\end{proof}

\begin{corollary}
	\label{cor: SEC limits}
	For $K, g > 0$ it holds that
	\[\lim_{a \rightarrow 0} x(a) = \lim_{a \rightarrow 0} x_{\textnormal{DF}}(a) = \lim_{a \rightarrow 0} x_{\textnormal{L}}(a) = 0\]
	and 
	\[ \lim_{a \rightarrow \infty} x(a) = \lim_{a \rightarrow \infty} x_{\textnormal{DF}}(a) = \lim_{a \rightarrow \infty} x_{\textnormal{L}}(a) = g \ .\]
\end{corollary}

\begin{proof}
    The limits 
    \[\lim\limits_{a \rightarrow \infty} x_{\textnormal{DF}}(a)  = g\ , \quad \lim\limits_{a \rightarrow \infty} x_{\textnormal{L}}(a) = g\ , \quad  \lim\limits_{a\rightarrow 0} x_{\textnormal{DF}}(a) = 0\quad \text{and}\quad \lim\limits_{a\rightarrow 0} x(a) = 0\]
    are obvious. Since $x(a) \leq x_{\textnormal{L}}(a) < x_{\textnormal{DF}}(a)$ for $a \leq g$ (Theorem \ref{thm: SEC model inequalities}), it follows that
    \[\lim\limits_{a \rightarrow 0} x_{\textnormal{L}}(a) = 0\ .\]
    In the same way, since $x_{\textnormal{L}}(a)\leq x(a) < x_{\textnormal{DF}}(a)$ for $a\geq g$ (Theorem \ref{thm: SEC model inequalities}), it follows that
    \[\lim\limits_{a\rightarrow \infty} x(a) = g\ .\]
\end{proof}

Before we discuss the dose-response behavior and the heuristics of the antibody depletion effect, let us briefly discuss the inequalities of Theorem \ref{thm: SEC model inequalities}. Among the inequalities, only $x(a) < x_{\textnormal{DF}}(a)$ has a simple reason: the effective concentration decreases with time in the depletion case, so more antibodies accumulate for the depletion-free case than for the depletion case. Unfortunately, there appears to be no deeper reason for the remaining inequalities. 

Although the lack of a dissociation term $-k_dx(t)$ implies a faster antibody accumulation a for the depletion-free accumulation IVP, compared to Langmuir kinetics, the accumulation rate hardly matters for the Langmuir isotherm. The accumulation model assumes a finite incubation time. On the other hand, the Langmuir isotherm is assumed to reach the equilibrium state, where even infinite incubation times are allowed. Thus, the accumulation rate does not suffice to motivate the inequality $x_{\textnormal{L}}(a) < x_{\textnormal{DF}}(a)$. Without calculations, it could equally well be assumed that the equilibrium surface concentration of bound antibodies was higher than the surface concentration of antibodies that could accumulate in the finite incubation time $\tau$.

The inequalities $x(a) < x_{\textnormal{L}}(a)$ for $a<g$ and $x_{\textnormal{L}}(a) < x(a)$ for $a>g$ might seem more intuitive, but verbal arguments easily fail. For sufficiently high antibody concentrations, it can be assumed that $x(a)\simeq x_{\textnormal{DF}}(a)$. However, it also holds that $x_{\textnormal{L}}(a)$ is a monotonically increasing function of the antibody concentration. It could very well have been the case that $x(a)< x_{\textnormal{L}}(a)$ remains valid up to the limit $a\rightarrow \infty$ where all models agree. In the same way, the depletion effect is more noticeable for low antibody concentration, because the binding of antibodies produces a stronger reduction of the relative antibody concentration if there are few antibodies to begin with. However, the Langmuir isotherm also decreases with decreasing antibody concentrations. Finally, there is no obvious reason for the equality $x(g) = x_{\textnormal{L}}(g)$.

\subsection{Dose-response behavior and heuristics for antibody depletion}
\label{subsec: heuristics for antibody depletion}

To discuss the dose-response behavior and to get an idea of the antibody depletion effect of the depletion accumulation model, we can consider the graphs for the suitable combinations of $K$ and $g$. Since all models scale with the epitope concentration $g$, it is convenient to consider the \textbf{fractional occupancies} $\theta(a) = \nicefrac{x(a)}{g}$ for the plots, which does not affect the relationship between the models.

\begin{figure}[h!]
	\centering
	\begin{subfigure}[c]{0.49\textwidth}
		\caption{Model plots linear scale}
		\includegraphics[width= \textwidth]{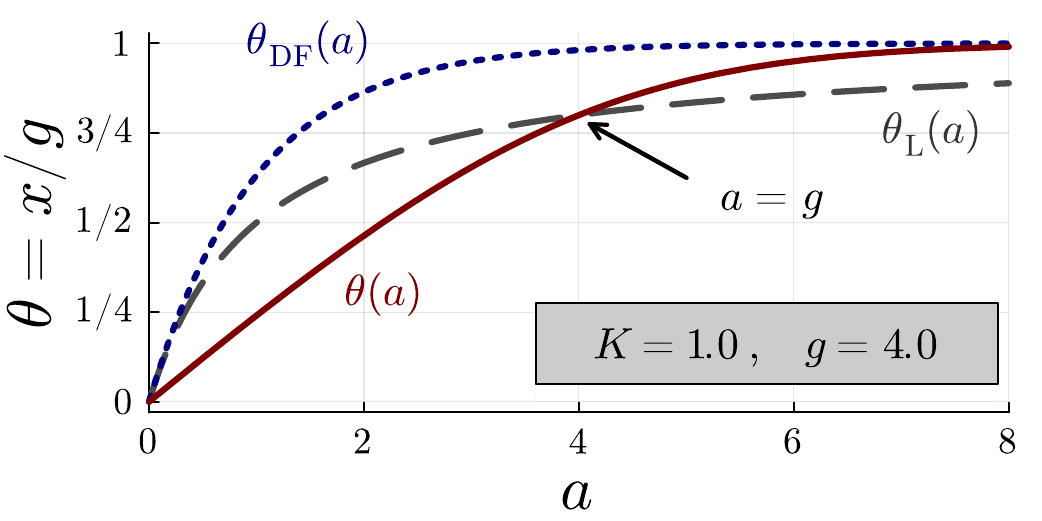}
	\end{subfigure}
	\begin{subfigure}[c]{0.49\textwidth}
		\caption{Model plots logarithmic scale}
		\includegraphics[width=\textwidth]{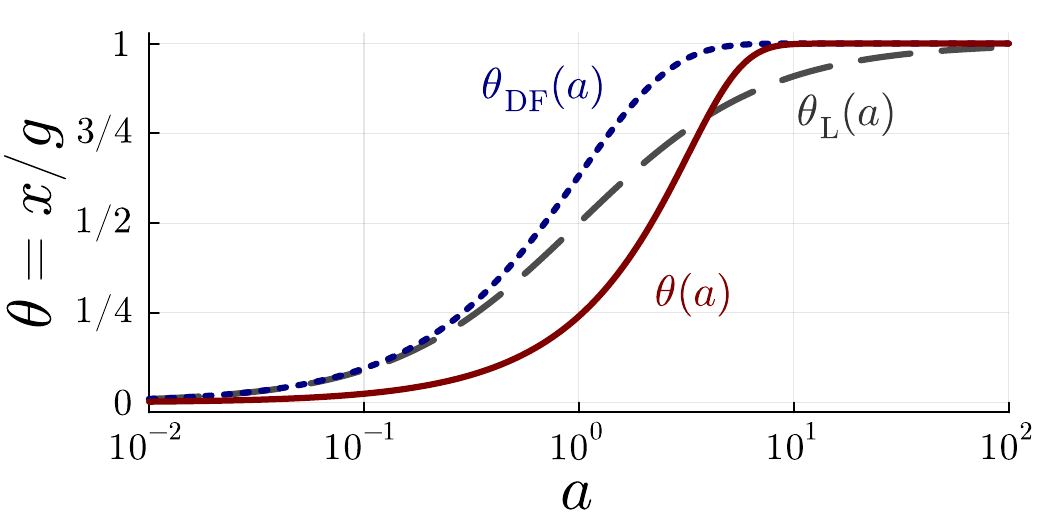}
	\end{subfigure}

	\caption{Plots for $K=1.0$ and $g = 4.0$ of $\theta(a)$ (red, solid line), $\theta_{\textnormal{DF}}(a)$ (blue, dotted line) and $\theta_{\textnormal{L}}(a)$ (gray, dashed line).} 
	\label{fig: model comparison 1}
\end{figure}

Figures \ref{fig: model comparison 1}a and \ref{fig: model comparison 1}b illustrate the inequalities of Theorem \ref{thm: SEC model inequalities} and the limiting behavior of Corollary \ref{cor: SEC limits} quite well. On the linear scale (Figure \ref{fig: model comparison 1}a), it can be seen that all models begin at $\theta =0$ for $a=0$. Until $a = g = 4$, the Langmuir isotherm $\theta_{\textnormal{L}}(a)$ is larger than the accumulation model $\theta(a)$. As stated in Theorem \ref{thm: SEC model inequalities} both models coincide at $a = g = 4$, from where the depletion accumulation model is larger than the Langmuir isotherm. Also, as stated in Theorem \ref{thm: SEC model inequalities}, the depletion-free accumulation model $\theta_{\textnormal{DF}}(a)$ is an upper bound for $\theta(a)$ and $\theta_{\textnormal{L}}(a)$ for all antibody concentrations. Finally, on the logarithmic scale (Figure \ref{fig: model comparison 1}b), it can be seen that all models converge to the same point for $a\rightarrow \infty$.

\begin{figure}[h!]
	\centering
	\begin{subfigure}[c]{0.49\textwidth}
		\caption{Model plots linear scale}
		\includegraphics[width= \textwidth]{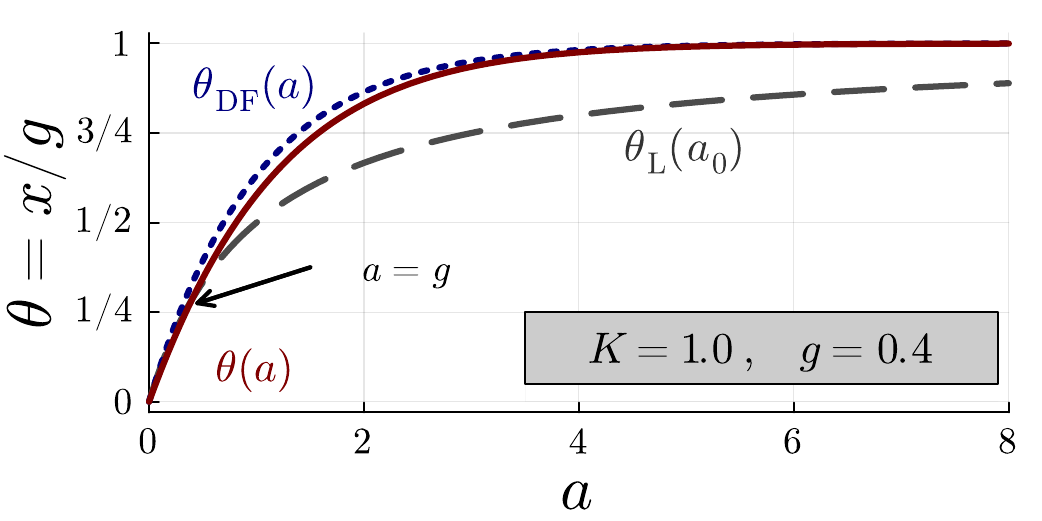}
	\end{subfigure}
	\begin{subfigure}[c]{0.49\textwidth}
		\caption{Model plots logarithmic scale}
		\includegraphics[width=\textwidth]{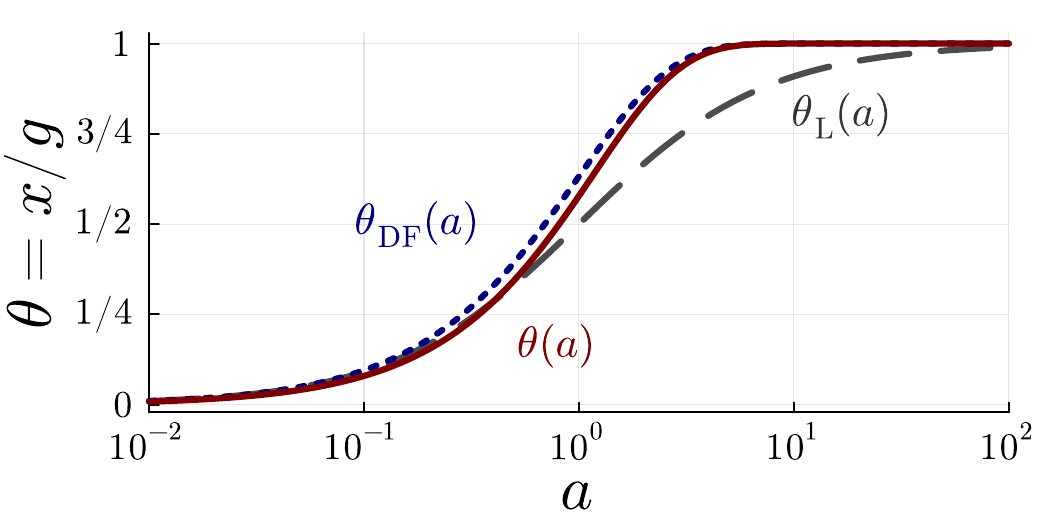}
	\end{subfigure}

	\caption{Plots for $K=1.0$ and $g = 0.4$ of $\theta(a)$ (red, solid line), $\theta_{\textnormal{DF}}(a)$ (blue, dotted line) and $\theta_{\textnormal{L}}(a)$ (gray, dashed line). } 
	\label{fig: model comparison 2}
\end{figure}

\begin{figure}[h!]
	\centering
	\begin{subfigure}[c]{0.49\textwidth}
		\caption{Model plots linear scale}
		\includegraphics[width= \textwidth]{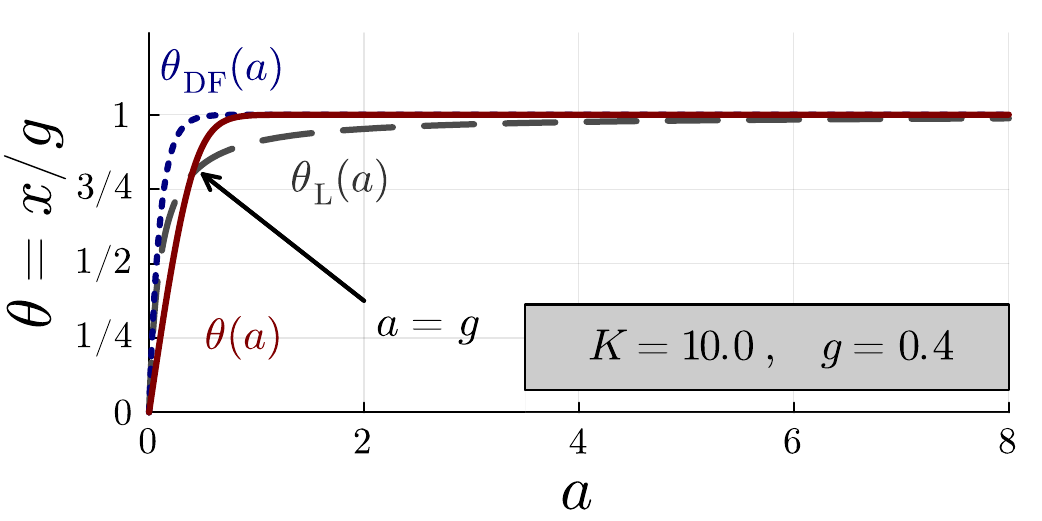}
	\end{subfigure}
	\begin{subfigure}[c]{0.49\textwidth}
		\caption{Model plots logarithmic scale}
		\includegraphics[width=\textwidth]{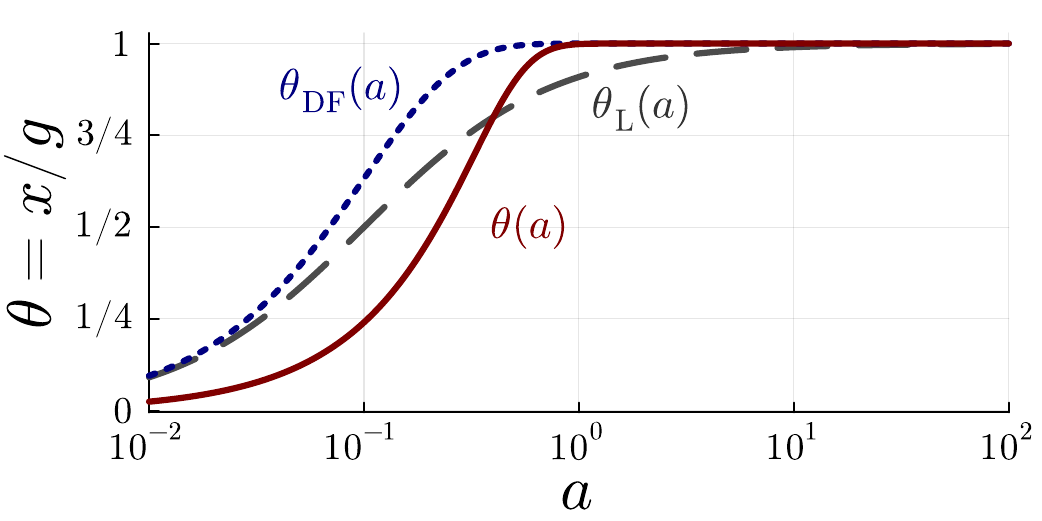}
	\end{subfigure}

	\caption{Plots for $K=10.0$ and $g = 0.4$ of $\theta(a)$ (red, solid line), $\theta_{\textnormal{DF}}(a)$ (blue, dotted line) and $\theta_{\textnormal{L}}(a)$ (gray, dashed line). } 
	\label{fig: model comparison 3}
\end{figure}

To get a better idea of the depletion behavior, different parameter configurations should be investigated. For this purpose, Figures \ref{fig: model comparison 2} and \ref{fig: model comparison 3} contain the same plots as Figure \ref{fig: model comparison 1} but with different values for $K$ and $g$. There are two important observations to be made. 

First, the graphs in Figure \ref{fig: model comparison 1}b are the same graphs as in Figure \ref{fig: model comparison 3}b, just shifted to the left. In the linear scale, the graphs are the same between Figures \ref{fig: model comparison 1}a and \ref{fig: model comparison 3}a up to a rescaling of the horizontal axis, because of the relationship between linear and logarithmic scales. In consequence, the depletion is the same in Figure \ref{fig: model comparison 1} ($K=1.0$ and $g=4.0$) and Figure \ref{fig: model comparison 3} ($K=10.0$ and $g = 0.4$). This suggests that the depletion effect depends on the product $Kg$.

The second observation is that the relationship between the Langmuir isotherm and the depletion-free accumulation model is the same in all plots. This is not an artifact of our choice for the system parameters $K$ and $g$, but a general principle. For both models, $g$ only acts as a scale parameter for the vertical axis and $K$ only acts as a scale parameter for the horizontal axis:
\[x_{\textnormal{DF}}(a) = g\cdot (1-e^{-(Ka)})\qquad \text{and}\qquad x_{\textnormal{DF}}(a) = g\cdot \frac{1}{1+\frac{1}{(K a)}}\ .\]
In fact, expressing the antibody concentration $a$ relative to the parameter $K$, i.e. $a = \nicefrac{b}{K}$, we obtain fractional occupancies whose shape does not depend on the system parameters $K$ and $g$:
\[\theta_{\textnormal{DF}}(b) = (1-e^{-b})\qquad \text{and}\qquad \theta_{\textnormal{DF}}(b) = \frac{1}{1+\frac{1}{b}}\ .\]
For the depletion accumulation model, the parameters $K$ and $g$ remain in the equation, but we observe that the shape only depends on the product $Kg$:
\[\theta(b) = \left\{ \begin{array}{ll}
	\frac{b(1-e^{b-Kg})}{Kg - be^{b-Kg}} & \ , \quad b \neq Kg\\[1em]
	\frac{1}{1+\frac{1}{Kg}} & \ , \quad b = Kg
\end{array}  \right. \ .\]
Thus, the depletion effect only depends on the product $Kg$. We can use these parameter-free fractional occupancies, together with Theorem \ref{thm: SEC model inequalities}, to analyze how $Kg$ affects the shape of the depletion. 

\begin{figure}[h!]
	\centering
	\begin{subfigure}[c]{0.49\textwidth}
		\caption{Model plots linear scale}
		\includegraphics[width= \textwidth]{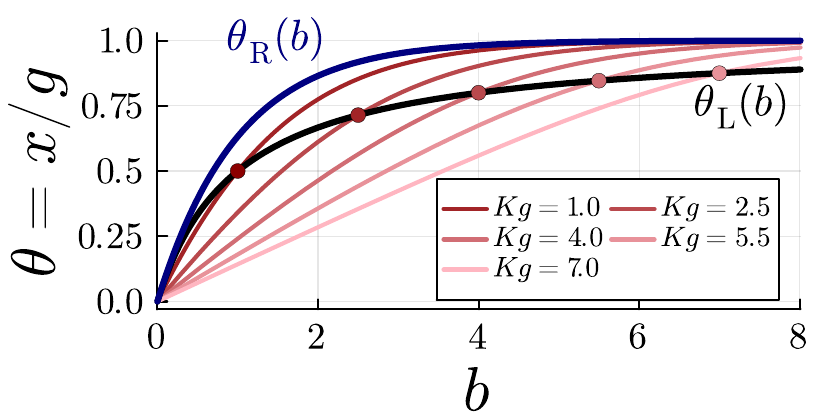}
	\end{subfigure}
	\begin{subfigure}[c]{0.49\textwidth}
		\caption{Model plots logarithmic scale}
		\includegraphics[width=\textwidth]{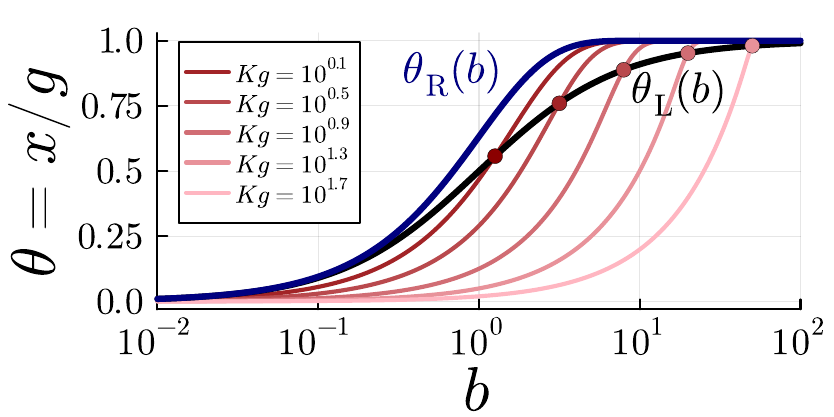}
	\end{subfigure}

	\caption{Plots of $\theta_{\textnormal{DF}}(b)$ and $\theta_{\textnormal{L}}(b)$ (which do not depend on $Kg$) and of $\theta(b)$ for several values of $Kg$, illustrating the $Kg$-dependence of the depletion shape. Note that $a = \frac{b}{K}$. So, the shapes of the depletion curves $\theta(b)$ for different values of $Kg$ should not be compared directly with each other, only with the shapes of $\theta_{\textnormal{DF}}(b)$ and $\theta_{\textnormal{L}}(b)$. } 
	\label{fig: model comparison 4}
\end{figure}

When the intersection point between the depletion accumulation model $\theta(b)$ and the Langmuir isotherm $\theta_{\textnormal{L}}(b)$ is further to the right, the depletion effect is more pronounced. Here, ``more pronounced'' means that the depletion effect covers a longer part of the non-trivial curvature of the depletion-free accumulation model. That is, the depletion accumulation model approximates the depletion-free accumulation model further to the right with respect to the point where the depletion-free accumulation model becomes almost constant (complete occupancy of epitopes). Since the intersection point is $a = g$, that is, $b = Kg$, we find that larger $Kg$ values lead to more pronounced depletion effects, which is illustrated in Figure \ref{fig: model comparison 4}.

\begin{remark}[\textbf{Depletion heuristic}]
    The larger $Kg$, the more pronounced the depletion effect.
\end{remark}

\section{Accumulation model for multiple epitope classes}

When there is more than one epitope class, the depletion accumulation IVP does not seem to have a general analytical solution. Nevertheless, we can show that there is a unique solution that satisfies most of the properties that we have found in the last section.

While we cannot find an analytical solution for the depletion accumulation IVP, the analytical solution of the depletion-free accumulation IVP is almost trivial. Because the antibody concentration is constant, the different epitope classes do not compete for available antibodies, becoming independent of each other:
\[\frac{d}{dt} x_{\textnormal{DF},i}(t) = k_{a,i} a (g_i -x_{\textnormal{DF},i}(t))\qquad \forall\ i \in \{1,\ldots,N\}\ .\]
Thus, the solution for each epitope class is given by the solution for a single epitope class:
\[x_{\textnormal{DF},i}(t) = g_i (1-e^{-k_{a,i} a t}) \qquad \forall\ i \in \{1,\ldots,N\}\ ,\]
\begin{equation}
	\label{eq: depletion-free multi-epitope solution}
	 \Rightarrow\qquad \boldsymbol{x}_{\textnormal{DF}}(t) = \begin{pmatrix}
	    g_1 (1-e^{-k_{a,1} a t})\\ \vdots \\ g_1 (1-e^{-k_{a,1} a t})
	\end{pmatrix}  \ .
\end{equation}
As for Theorem \ref{thm: unique solution single epitope class}, uniqueness of this solution for $t\geq 0$ follows from a general existence and uniqueness theorem \cite[Chapter III, \S10, Theorem VI]{Walter_1998}.

\subsection{Properties of the multi-epitope accumulation IVP solution}
\label{subsec: Properties of the multi-epitope accumulation IVP solution}

Because of the general existence and uniqueness theorem \cite[Chapter III, \S10, Theorem VI]{Walter_1998}, we know that there is a unique solution for the depletion accumulation IVP around $t=0$. Let $\xi > 0$ denote the maximal time for which the unique solution $\boldsymbol{x}\colon [0,\xi)\rightarrow\mathbb{R}^N$ exists. Of course, this \textbf{maximal existence time} depends on the initial antibody concentration $a$ and the epitope classes $\{(g_i,k_{a,i})\}_{i=1}^N$, which we suppress in the notation.

Eventually, we want to prove that $\xi = \infty$. However, since we do not have an analytical solution for the depletion accumulation IVP if there is more than one epitope class, we need to prove additional properties of the unknown, locally unique solution first.

\begin{theorem}
	\label{theorem: integral equations}
	Let $\{(g_i,k_{a,i})\}_{i=1}^N$  be epitope classes, $a>0$ and let $\xi > 0$ be the corresponding maximal existence time. A vector-valued function $\boldsymbol{v}\colon [0,\xi)\rightarrow \mathbb{R}^N$ with $\boldsymbol{v}(0) = 0$ is the solution of the depletion accumulation IVP if and only if its components satisfy the following integral equations
	\[v_i(t) = g_i\left(1-e^{-k_{a,i}\int_0^t \left(a - \sum_{j=1}^N v_j(s)\right)\  ds}\right)\qquad \forall\  i\in \{1,\ldots, N\}\ .\]
\end{theorem}

\begin{proof}
    Let $\boldsymbol{v}\colon [0,\xi) \rightarrow \mathbb{R}^N$ be a vector-valued function with $\boldsymbol{v}(0) = 0$ that satisfies the integral equations
    \[v_i(t) = g_i\left(1-e^{-k_{a,i}\int_0^t \left(a - \sum_{j=1}^N v_j(s)\right)\  ds}\right)\qquad \forall\  i\in \{1,\ldots, N\}\ .\]
    For the time derivatives of the components we calculate
    \begin{align*}
    	\frac{d}{dt} v_i(t) &= k_{a,i}g_i e^{-k_{a,i}\int_0^t \left(a - \sum_{j=1}^N v_j(s)\right)   ds}\quad  \frac{d}{dt} \int_0^t \left(a - \sum_{j=1}^N v_j(s)\right)  ds \\[1em]
    	&= k_{a,i}\left(a-\sum_{j=1}^N v_j(t)\right) g_i e^{-k_{a,i}\int_0^t \left(a - \sum_{j=1}^N v_j(s)\right)   ds}\\[1em]
    	&= k_{a,i}\left(a-\sum_{j=1}^N v_j(t)\right)\left(g_i - g_i + g_i e^{-k_{a,i}\int_0^t \left(a - \sum_{j=1}^N v_j(s)\right)  ds}\right)\\[1em]
    	&= k_{a,i}\left(a-\sum_{j=1}^N v_j(t)\right)\left(g_i - g_i \left(1- e^{-k_{a,i}\int_0^t \left(a - \sum_{j=1}^N v_j(s)\right)  ds}\right)\right)\\[1em]
    	&= k_{a,i}\left(a-\sum_{j=1}^N v_j(t)\right)\left(g_i- v_{i}(t)\right)\quad ,
    \end{align*}
    where we have used $v_i(t) = g_i\big(1-e^{-k_{a,i}\int_0^t \left(a - \sum_{j=1}^N v_j(s)\right)\  ds}\big)$ in the last step. Thus, $\boldsymbol{v}$ solves the depletion accumulation IVP.

    Let now $\boldsymbol{y}\colon [0,\xi)\rightarrow \mathbb{R}^N$ solve the multi-epitope accumulation IVP. Due to the local uniqueness of the solution, it follows that $\boldsymbol{y}(t) =  \boldsymbol{v}(t)$ for all $t\in [0,\xi)$. But then $\boldsymbol{y}$ also satisfies the integral equations.
\end{proof}

The integral equations have a simple and intuitive interpretation. Without antibody depletion, the concentration of bound antibodies for the individual epitope classes is given by
\[x_{\textnormal{DF},i}(t) =  g_i(1-e^{-k_{a,i}a t})\ .\]
Since the antibody concentration is constant, we can express the term $a t$ as time integral 
\[a t = \int_0^t a\ ds\ .\]
We may understand the integral as a cumulative effect of the antibody concentration over time. In the depletion case, the antibody concentration is given by $a(t) = a -\sum_{j=1}^N x_j(t)$ and the cumulative antibody concentration effect is
\[\int_0^t \alpha(s)\ ds = \int_0^t \left(a -\sum_{j=1}^N x_j(s)\right)\ ds\ .\]
This is precisely the term in the exponent of the integral equations of Theorem \ref{theorem: integral equations}. Thus, the depletion-free and the depletion cases behave similarly. The only difference is that the antibody concentration varies over time in the depletion case. For later use, we may already define here the cumulative antibody concentration.

\begin{definition}
    \label{def: cumulative antibody concentration}
    Let $\alpha(t\ ; a)$ denote the antibody concentration at time $t$ for the initial antibody concentration $a$, then we define the \textbf{cumulative antibody concentration} as
    \[A(t\ ; a) = \int_{0}^t \alpha(s\ ; a)\  ds\ .\]
    In particular, we call $\alpha(t\ ; a) = a - \sum_{j=1}^N x_j(t)$ the \textbf{depletion antibody concentration} and the corresponding $A(t\ ; a)$ the \textbf{cumulative depletion antibody concentration}.
\end{definition}

Given the similarity between the solution of the depletion-free accumulation model and the integral equations, it is reasonable to assume that also the solution of the depletion accumulation model is increasing monotonically. However, to prove this property, we first need the following lemma.

\begin{lemma}
	\label{lemma: all derivatives must be zero}
	Let $\{(g_i , k_{a,i})\}_{i=1}^N$ be epitope classes, let $a>0$ and let $\xi > 0$ be the corresponding maximal existence time. Furthermore, let $\boldsymbol{x}\colon [0,\xi)\rightarrow \mathbb{R}^N$ denote the unique solution of the depletion accumulation IVP. Then it holds for all $t\in [0,\xi)$ that
	\begin{enumerate}
		\item $\displaystyle \quad \exists\  i \in \{1,\ldots,N\} \colon  \frac{d}{dt} x_i(t) = 0\qquad \Leftrightarrow\qquad \frac{d}{dt} x_i(t) = 0 \quad \forall \ i \in \{1,\ldots, N\} $
		\item $\displaystyle \quad \exists\  i \in \{1,\ldots,N\} \colon  \frac{d}{dt} x_i(t) > 0\qquad \Leftrightarrow\qquad \frac{d}{dt} x_i(t) > 0 \quad \forall \ i \in \{1,\ldots, N\} $
	\end{enumerate}
\end{lemma}

\begin{proof}
    We only show the first case, $\frac{d}{dt}x_i(t)=0$, as all steps and arguments will be the same for the second case $\frac{d}{dt}x_i(t)>0$.
    
    First, we observe that the direction ``$\Leftarrow$'' is trivial. For the opposite direction, assume that $\frac{d}{dt}x_i(t) = 0$ for an arbitrary $i \in \{1,\ldots,N\}$. Because of Theorem \ref{theorem: integral equations}, the function $x_i$ satisfies the integral equation
    \[x_i(t) = g_i\left(1-e^{-k_{a,i}\int_0^t \left(a - \sum_{j=1}^N x_j(s)\right)  ds}\right)\ .\]
    In the proof of Theorem \ref{theorem: integral equations}, we have already calculated that we can express the derivative as
    \[\frac{d}{dt} x_i(t) = k_{a,i}\left(a-\sum_{j=1} x_j(t)\right) g_i e^{-k_{a,i}\int_0^t \left( a - \sum_{j=1}^N x_j(s)\right)  ds}\ .\]
    Since all $x_i(t)$ are differentiable, and thus continuous on $[0,\xi)$, it follows that ${s \mapsto x_i(s)}$  is bounded on $[0,t]$. So, $\displaystyle \int_0^t \left( a -\sum_{j=1}^N x_i(s)\right) ds$ is finite and $e^{\int_0^t \left( a -\sum_{j=1}^N x_i(s)\right) ds}$ is positive.  By Definition \ref{def: accumulation IVPs},  all $g_i$ and $k_{a,i}$ are also positive. Hence, the condition $\frac{d}{dt}x_i(t)= 0$ leads to
    \[0 = \frac{d}{dt} x_i(t) = k_{a,i}\left(a-\sum_{j=1} x_j(t)\right) g_i e^{-k_{a,i}\int_0^t \left( a - \sum_{j=1}^N x_j(s)\right)  ds}\ ,\]
    \[\Rightarrow \qquad 0 =  a -\sum_{j=1}^N x_j(t) \ .\]
    Finally, we note that $0 =  a -\sum_{j=1}^N x_j(t)$ is not a condition for a specific $i$, but applies to all $i\in \{1,\ldots, N\}$. That is, let $i \in \{1,\ldots,N\}$ be arbitrary. Then, since $0 =  a -\sum_{j=1}^N x_j(t)$, it follows that
    \[\frac{d}{dt} x_i(t) = k_{a,i}\left(a-\sum_{j=1} x_j(t)\right) g_i e^{-k_{a,i}\int_0^t \left(a - \sum_{j=1}^N x_j(s)\right)  ds} = 0\ .\]
    Thus, $\frac{d}{dt}x_i(t)= 0$ for all $i\in \{1,\ldots,N\}$, which concludes the proof.
\end{proof}

With this lemma, we can show that the concentration of bound antibodies increases monotonically for all epitope classes.

\begin{theorem}
	\label{theorem: positive derivatives}
	Let $\{(g_i , k_{a,i})\}_{i=1}^N$ be epitope classes, let $a> 0$ and let $\xi > 0$ be the corresponding maximal existence time. Furthermore, let $\boldsymbol{x}\colon [0,\xi)\rightarrow \mathbb{R}^N$ denote the unique solution of the depletion accumulation IVP. Then it holds that
	\[\frac{d}{dt} x_i(t) > 0\qquad \forall\  t \in [0,\xi) \qquad \forall \ i \in \{1,\ldots, N\}\ .\]
\end{theorem}

\begin{proof}
    From the differential equations of the depletion accumulation IVP
    \[\frac{d}{dt} x_i(t) = k_{a,i} \left(a- \sum_{j=1}^N x_j(t)\right)(g_i - x_i(t)) \qquad \forall\ i\in \{1,\ldots, N\}\]
    and the initial value $\boldsymbol{x}(0) = 0$, it follows that
    \[\frac{dx_i}{dt}(0) = k_{a,i} a g_i  > 0 \qquad \forall \ i \in \{1,\ldots, N\}\ . \]
    Assume now, that there is a $t_- \in (0,\xi)$ such that $\frac{dx_i}{dt}(t_-) \leq 0$ for an arbitrary $i\in \{1,\ldots, N\}$. Because of the continuity of $\frac{d}{dt}x_i(t)$, which follows from the continuity of the right-hand side of the differential equations, there must then be a $0 < t_0\leq t_-$ for which $\frac{dx_i}{dt}(t_0) = 0$. By lemma \ref{lemma: all derivatives must be zero}, this implies that all $\frac{dx_i}{dt}(t_0) = 0$, i.e. $\frac{d\boldsymbol{x}}{dt}(t_0) = 0$.

    Since the depletion accumulation IVP is an autonomous system with a locally Lipschitz continuous map, general properties of autonomous systems \cite[Chapter III, \S 10, XI]{Walter_1998} would imply that $\boldsymbol{x}(t) \equiv 0$. Yet, $\frac{dx_i}{dt}(0) > 0$ implies that there is an $\varepsilon > 0$ such that $x_i(\varepsilon) > 0$. This is a contradiction. Hence, there is no $t_0$ for which any $\frac{dx_i}{dt}(t_0) = 0$ and thus $\frac{d}{dt}x_i(t) > 0$ for all $t\in [0,\xi)$ for all $i \in \{1,\ldots, N\}$.
    
\end{proof}

We can use the last theorem to show that the unique solution of the depletion accumulation IVP is bounded. Before we can do that, however, we need a simple result from calculus.

\begin{lemma}
	\label{lemma: larger derivative and larger function}
	Let $f\colon I \longrightarrow \mathbb{R}$ be a differentiable function on a non-empty interval $I$. If 
	\[f'(t)\geq 0 \qquad \forall\  t\in I\qquad \text{and}\qquad \exists\  a\in I\colon f(a)\geq 0\]
	then it holds that $f(t) \geq 0$ for all $t\in I$ with $t\geq a$. If the inequalities are strict
	\[f'(t) > 0 \qquad \text{and}\qquad f(a)> 0\]
	then the strict inequality $f(t)>0$  holds for all $t\in I$ with $t \geq a$.
\end{lemma}
\begin{proof}
    This is a straightforward consequence of the mean value theorem. 
\end{proof}

Now, we can prove the bounds for the unique solution of the depletion accumulation IVP.

\begin{theorem}
	\label{theorem: constant bounds for multi-epitope solution}
	Let $\{(g_i, k_{a,i})\}_{i=1}^N$ be epitope classes, let $a>0$ and let $\xi\geq 0$ be the corresponding maximal existence time. Furthermore, let $\boldsymbol{x}\colon [0,\xi)\rightarrow \mathbb{R}^N$ denote the solution of the multi-epitope accumulation IVP. Then it holds for all $t\in [0,\xi)$ and for all $i\in \{1,\ldots, N\}$ that
	\[0\leq x_i(t)< \min\{a,g_i\}\qquad \text{and}\qquad \sum_{j=1}^N x_j(t)< a\ .\]
\end{theorem}

\begin{proof}
    Since $\boldsymbol{x}$ solves the multi-epitope accumulation IVP, it holds that $\boldsymbol{x}(0) = 0$. Furthermore, according to theorem \ref{theorem: positive derivatives} it holds that $\frac{d}{dt}x_i(t)>0$ for all $t\in [0,\xi)$ and for all $i\in\{1,\ldots,N\}$. Thus, because of Lemma \ref{lemma: larger derivative and larger function} it follows that
    \[0\leq x_i(t)\qquad \forall\ t\in [0,\xi)\ ,\quad \forall\ i \in \{1,\ldots,N\}\ . \]
    
    By Theorem \ref{theorem: integral equations} and the corresponding proof it holds that
    \[\frac{d}{dt} x_i(t) = k_{a,i}\left(a-\sum_{j=1}^N x_j(t)\right) g_i e^{-k_{a,i}\int_0^t \left(a - \sum_{j=1}^N x_j(s)\right)  ds}\qquad \forall\ i \in \{1,\ldots,N\}\ .\]
    Since $\frac{d}{dt}x_i(t) >0$, according to theorem \ref{theorem: positive derivatives}, it follows that
    \[0< a - \sum_{j=1}^N x_j(t) \quad \Rightarrow \quad  \sum_{j=1}^N x_j(t) < a \ .\]
    Note that we have already proven that the sum consists only of non-negative values, i.e. $0\leq x_i(t)$. So, it also follows that
    \[x_i(t)< a \quad \forall \ i\in \{1,\ldots,N\}\ .\]
    
    It remains to show that $x_i(t) < g_i$. Since $\boldsymbol{x}$ solves the multi-epitope accumulation IVP, it holds that
    \[\frac{d}{dt}x_i(t) = k_{a,i}\left(a - \sum_{j=1}^N x_j(t)\right)(g_i - x_i(t))\qquad \forall\ i\in \{1,\ldots,N\}\ .\]
    Since $\frac{d}{dt}x_i(t) > 0$ and since we have already shown that $a - \sum_{j=1}^N x_j(t)> 0$, it follows that $g_i-x_i(t)> 0$, which is equivalent to $x_i(t)<g_i$, for all $i\in \{1,\ldots,N\}$. 
\end{proof}

With this theorem, we can derive maybe the most important result of this section as a simple corollary.

\begin{corollary}[\textbf{Existence of global solution}]
	\label{cor: multi-epitope global existence and uniqueness}
    For all epitope class configurations $\{(g_i, k_{a,i})\}_{i=1}^N$ and all $a>0$, there exists a unique solution $\boldsymbol{x}\colon [0,\infty) \rightarrow \mathbb{R}^N$ for the depletion  accumulation IVP. That is, the maximal existence time is $\xi = \infty$.
\end{corollary}
\begin{proof}
    Because of Theorem \ref{theorem: constant bounds for multi-epitope solution}, the half-trajectory $\boldsymbol{x}([0,\xi))$ is contained in a compact set of $\mathbb{R}^N$. The rest follows from \cite[Chapter III, \S 10, XI]{Walter_1998} since the depletion accumulation IVP is a locally Lipschitz continuous autonomous system.
\end{proof}

A second useful property can easily be derived with Theorem \ref{theorem: positive derivatives} and Theorem \ref{theorem: constant bounds for multi-epitope solution}.

\begin{corollary}
	\label{cor: concave accumulation}
    Let $\boldsymbol{x}\colon [0,\infty)\rightarrow \mathbb{R}^N$ be the unique solution of a depletion accumulation IVP. Then the components $x_i\colon [0,\infty)\rightarrow \mathbb{R}$ are strictly concave.
\end{corollary}

\begin{proof}
As a solution of the depletion accumulation IVP, the components of $\boldsymbol{x}$ satisfy.
\[\frac{d}{dt}x_i(t) = k_{a,i}\left(a -\sum_{j=1}^N x_j(t)\right)(g_i-x_i(t))\ .\]
From Theorem \ref{theorem: constant bounds for multi-epitope solution} it follows that $a-\sum_{j=1}^N x_j(t) \geq 0$ and $g_i -x_i(t) \geq 0$, such that $\frac{d}{dt}x_i(t)$ decreases as the $x_j(t)$ increase. Finally, since $\frac{d}{dt}x_i(t)>0$ by Theorem \ref{theorem: positive derivatives}, all $x_i(t)$ are strictly monotonically increasing, which means that all $\frac{d}{dt}x_i(t)$ are strictly monotonically decreasing.
\end{proof}

For a single epitope class, we had closed-form expressions, allowing us to simply calculate that the natural dose-response property is satisfied. That is, increasing the initial antibody concentration must not decrease the amount of bound antibodies. Since we do not have a close-form expression of the solution for multiple epitope classes, we need to derive the dose-response property from the properties that we have derived so far in this section. For this, we need the following lemma on the antibody concentration at time $t$.

\begin{lemma}
	\label{lemma: antibody concentration inequalities}
    Let $\boldsymbol{x}(t\ ; a) = \boldsymbol{x}(t\ ; a, \{(g_i,k_{a,i})\}_{i=1}^N)$ denote the solution of a depletion accumulation IVP for the initial antibody concentration $a$ and let $\alpha(t\ ; a)$ denote the depletion antibody concentration at time $t$ for the initial antibody concentration $a$. If $b > a$ it holds that $\alpha(t\ ; b)> \alpha(t\ ; a)$ and $A(t\ ; b) > A(t\ ; a)$ for all $t\in [0,\infty)$.
\end{lemma}

\begin{proof}
    First, we recall form Definition \ref{def: cumulative antibody concentration} that $A(t\ ; a) = \int_0^t \alpha(s\ ; a)\ ds$. Thus, if we can show that $\alpha(t\ ; b)\geq \alpha(t\ ; a)$ for all $t\in [0,\infty)$, the inequality $A(t\ ; b)\geq \alpha(t_0\ ; a)$ follows from the properties of integration. Hence, we will focus on $\alpha(t\ ; b) \geq \alpha(t\ ; a)$.
    
    Since $\boldsymbol{x}(0\ ; a) = 0$ and $b > a$, it holds that
    \[\alpha(0\ ; a) = a < b = \alpha(0\ ; b) \qquad\Rightarrow \qquad \alpha(0\ ; b)-\alpha(0\ ; a)> 0\ .\]
    Since the $x_i(t)$ are differentiable, so are $\alpha(t\ ; a)$, $\alpha(t\ ; b)$ and thus $\alpha(t\ ; b)-\alpha(t\ ; a)$. Our goal is now to show that $\alpha(t\ ; b)-\alpha(t\ ; a) > 0$ for all $t\in [0,\infty)$. 
    
    Let us assume that there is a $t_{-}\in [0,\infty)$ such that $\alpha(t_-\ ; b)-\alpha(t_-\ ; a)\leq 0$. Because of the continuity of $\alpha(t\ ; b)-\alpha(t\ ; a)$, there must be a smallest $t_0> 0$ such that $\alpha(t_0\ ; b)-\alpha(t_0\ ; a) = 0$. If we can show that no such $t_0$ exists, then we will have proven that $\alpha(t\ ; b)-\alpha(t\ ; a) > 0$ for all $t\in [0,\infty)$.

    Since $t_0 > 0$ and since by assumption $t_0$ is the smallest value such that $\alpha(t_0\ ; b)-\alpha(t_0\ ; a) = 0$, it holds that $\alpha(t\ ; b)-\alpha(t\ ; a)> 0$ for all $t\in [0,t_0)$. Thus,
    \[A(t\ ; b) = \int_0^t \alpha(s\ ; b)\  ds > \int_0^t \alpha(s\ ; a)\ ds = A(t\ ; a)\quad \forall t\in [0,t_0]\ .\]
    This also implies that $1-e^{-k_{a,i}A(t\ ; b)}> 1-e^{-k_{a,i}A(t\ ; a)}$. Because of Theorem \ref{theorem: integral equations} and Definition \ref{def: cumulative antibody concentration}, it follows that
    \begin{align*}
    	x_i(t\ ; b) &= g_i(1-e^{-k_{a,i}\int_0^t a -\sum_{j=1}^N x_j(s\ ;b)\ ds}) = g_i(1-e^{-k_{a,i}A(t\ ; b)})\\[1em]
    	&> g_i(1-e^{-k_{a,i}A(t\ ; a)}) = x_i(t\ ;a)
    \end{align*}
    for all $t\in [0,t_0]$ and for all $i\in \{1,\ldots, N\}$. This implies that
    \[\Gamma \coloneqq \min_{t\in [0,t_0], i \in \{1,\ldots,N\}}\  x_i(t\ ;b) - x_i(t\ ;a) > 0 \ .\]
    Because of the bounds from Theorem \ref{theorem: constant bounds for multi-epitope solution}, it holds that $0\leq x_i(t\ ;b)< g_i$ for all $t\in [0,\infty)$ and for all $i\in\{1,\ldots,N\}$. Furthermore, the theorem states that $a - \sum_{j=1}^N x_j(t\ ;a) > 0$, such that
    \[\Psi \coloneqq \max_{t\in [0,t_0], i \in \{1,\ldots,N\}}\  \frac{1}{g_i - x_i(t\ ;b)} > 0\]
    \[\text{and}\qquad \Phi \coloneqq \min_{t\in [0,t_0]} \alpha(t\ ; a) = \min_{t\in [0,t_0]}  a -\sum_{j=1}^N x_j(t\ ;a) > 0\ . \]
    In summary, we have
    \[\Gamma \Psi \Phi > 0 \ .\]
    Since $\alpha(t\ ; b)-\alpha(t\ ; a)$ is continuous and, by assumption, $t_0$ is the smallest value with $\alpha(t_0\ ; b)-\alpha(t_0\ ; a) = 0$, there is an $\varepsilon > 0$ such that
    \[\alpha(t\ ; b)- \alpha(t\ ; a) < \Gamma \Psi \Phi\qquad \forall\ t\in (t_0-\varepsilon,t_0]\ . \]
    \[\Leftrightarrow \qquad \alpha(t\ ; b)  < \Gamma \Psi \Phi + \alpha(t\ ; a) \qquad \forall\ t\in (t_0-\varepsilon,t_0]\ . \]
    Now, since $\boldsymbol{x}$ is the solution of the multi-epitope accumulation IVP it follows for all $i\in \{1,\ldots,N\}$ and all $t\in (t_0-\varepsilon,t_0]$ that
    \begin{align*}
    	\frac{1}{k_{a,i}}&\left( \frac{d}{dt}x_i(t\ ;a) - \frac{d}{dt}x_i(t\ ;b) \right) \\[1em] 
    	&=\left(a - \sum_{j= 1}^N x_j(t\ ; a)\right)(g_i-x_i(t\ ;a)) - \left(b - \sum_{j= 1}^N x_j(t\ ; b)\right)(g_i-x_i(t\ ;b))\\[1em]
    	&= \alpha(t\ ; a)(g_i -x_i(t\ ;a)) - \alpha(t\ ; b)(g_i -x_i(t\ ; b))\\[1em]
    	& > \alpha(t\ ; a)(g_i -x_i(t\ ;a)) - (\alpha(t\ ; a) + \Gamma\Psi\Phi)(g_i -x_i(t\ ; b))\\[1em]
    	&\quad = \alpha(t\ ; a)(x_i(t\ ; b) - x_i(t\ ; a)) - \Phi \Gamma \Psi (g_i-x_i(t\ ; b))\\[1em]
    	&\quad \geq \alpha(t\ ; a)(x_i(t\ ; b) - x_i(t\ ; a)) - \alpha(t\ ; a)(x_i(t\ ; b) - x_i(t\ ; a)) \cdot 1 = 0\ .
    \end{align*}
    Note that we have used the definitions of $\Gamma$, $\Psi$ and $\Phi$ in the last step. With the last calculation, we have obtained the statement
    \[\frac{d}{dt}x_i(t\ ;a) - \frac{d}{dt}x_i(t\ ;b) > 0\qquad \forall i \in \{1,\ldots,N\}\ ,\quad \forall t \in (t_0-\varepsilon,t_0]\ .\]
    For the time-derivatives of $\alpha(t\ ; b)$ and $\alpha(t\ ; a)$ this means that
    \[\frac{d}{dt}\alpha(t\ ; b) = - \sum_{j=1}^N \frac{d}{dt}x_j(t\ ;b) > - \sum_{j=1}^N\frac{d}{dt}x(t\ ;a) = \frac{d}{dt}\alpha(t\ ; a)\qquad \forall \  t\in (t_0-\varepsilon,t_0]\]
    \[\Leftrightarrow \qquad \frac{d}{dt}\alpha(t\ ; b) - \frac{d}{dt}\alpha(t\ ; a) > 0 \qquad \forall\ t\in (t_0-\varepsilon,t_0]\ .\]
    However, by definition of $t_0$, it also holds that $\alpha(t_0 - \delta\ ; b)-\alpha(t_0-\delta\ ; a) > 0$ for $0<\delta<\varepsilon.$ Since $\frac{d}{dt}\alpha(t\ ; b) - \frac{d}{dt}\alpha(t\ ; a) > 0$ for $t\in (t_0-\varepsilon,t_0]$ and because of Lemma \ref{lemma: larger derivative and larger function}, it follows that $\alpha(t_0\ ; b)- \alpha(t_0\ ; a) > 0$. This contradicts the definition of $t_0$ to be the smallest value such that $\alpha(t_0\ ; b) - \alpha(t_0\ ; a) = 0$. Hence, there is no finite $t_0$ for which $\alpha(t_0\ ; b)-\alpha(t_0\ ; a) = 0$, which concludes the proof. 
\end{proof}

Having derived the properties of the antibody concentration, we can prove the natural dose-response property.

\begin{theorem}
\label{theorem: larger antibody concentrations}
Let $\boldsymbol{x}(t\ ; a) = \boldsymbol{x}(t\ ; a, \{(g_i,k_{a,i})\}_{i=1}^N)$ denote the solution of the depletion accumulation IVP for the initial antibody concentration $a$ and the epitope classes $\{(g_i,k_{a,i})\}_{i=1}^N$. If $a < b$, then it holds that
\[x_i(t\ ; a) <  x_i(t\ ; b)\qquad \forall\ i \in\{1,\ldots,N\}\ , \quad \forall t \in [0,\infty).\]
Furthermore, if $k_{a,i} < k_{a,j}$ it holds that
\[\frac{x_i(t\ ; a)}{g_i} < \frac{x_j(t\ ; a)}{g_j}\qquad \forall \ t\in [0,\infty)\ .\]
\end{theorem}

\begin{proof}
    Because of Theorem \ref{theorem: integral equations} and Definition \ref{def: cumulative antibody concentration} it holds that
    \[x_i(t\ ;a) = g_i(1-e^{-k_{a,i}A(t\ ; a)})\ .\]
    According to Lemma \ref{lemma: antibody concentration inequalities}, we have $A(t\ ; a) < A(t\ ; b)$ for $a < b$. Thus,
    \[x_i(t\ ;a) = g_i(1-e^{-k_{a,i}A(t\ ; a)}) < g_i(1-e^{-k_{a,i}A(t\ ; b)}) = x_i(t\ ;b)\ .\]
    For the second part of the theorem, we observe that $k_{a,i} <  k_{a,j}$ implies that
    \[\frac{x_i(t\ ;a)}{g_i} = 1-e^{-k_{a,i}A(t\ ; a)} < 1-e^{-k_{a,j}A(t\ ; a)} = \frac{x_j(t\ ;a)}{g_j}\ .\]
\end{proof}

\subsection{Dose-response behavior and approximations}
As before, with the solution of the depletion-free accumulation IVP \eqref{eq: depletion-free multi-epitope solution} we have also obtained the depletion-free accumulation model and the corresponding dose-response curve
\begin{equation}
\label{eq: depletion-free accumulation model}
\boldsymbol{x}_{\textnormal{DF}}(\tau\ ; a) = \begin{pmatrix} g_1 (1-e^{-k_{a,1} a \tau})\\ \vdots \\ g_1 (1-e^{-k_{a,1} a \tau}) \end{pmatrix}\qquad \Rightarrow \qquad X_{\textnormal{DF}}(a) = \sum_{i=1}^N g_i (1-e^{-k_{a,i}}a \tau)\ . 
\end{equation}
However, we have no closed-form expression for the solution of the depletion accumulation IVP and thus no expressions for the depletion accumulation model and the corresponding dose-response curve. At this point, we only know about their existence for all $\tau\geq 0$, because of Corollary \ref{cor: multi-epitope global existence and uniqueness}. But, we can use the derived properties to characterize the depletion accumulation model. The key idea will be to consider the cumulative depletion antibody concentration $A(\tau\ ; a)$, which determines the depletion accumulation model according to Theorem \ref{theorem: integral equations} and Definition \ref{def: cumulative antibody concentration}:
\[x_i(\tau\ ;a) = g_i(1-e^{-k_{a,i}A(t\ ; a)})\ .\]

We can derive simple bounds for the multi-epitope accumulation model $X(a)$ if we find an $a_*(a)$ with $0\leq a_*(a) \leq  \alpha(t\ ; a)$ for all $t\in [0,\tau]$, as then
\[a_*(a) \tau \leq A(\tau\ ; a)\leq a \tau\ .\]

\begin{lemma}
\label{lemma: weak bounds multi-epitope accumulation model}
Let $a_*(a) \coloneqq \max\{0\ ; a - X_{\textnormal{DF}}(a)\}$ then it holds that
\[X_{\textnormal{DF}}(a_*(a))\leq X(a)\leq X_{\textnormal{DF}}(a)\qquad \forall \  a > 0\ .\]
\end{lemma}

\begin{proof}
    Let $\alpha(t\ ; a)$ be the depletion accumulation concentration at time $t$. Because of Theorem \ref{theorem: constant bounds for multi-epitope solution}, it holds that $\alpha(t\ ; a) \leq a$ for all $t \geq 0$ and thus by Definition \ref{def: cumulative antibody concentration} it also holds that $A(t\ ; a) \leq a t$ for all $t\geq 0$. From Theorem \ref{theorem: integral equations} it then follows that
    \[x_i(t\ ; a)  = g_i(1-e^{-k_{a,i}A(t\ ; a)}) \leq g_i(1-e^{-k_{a,i}a t}) = x_{\textnormal{DF},i}(t\ ; a)\]
    for all $a> 0$, all $t \geq 0$ and all $i \in \{1,\ldots, N\}$. Setting $t=\tau$ proves the upper bound $X(a) \leq X_{\textnormal{DF}}(a)$.
    
    But the upper bound also proves that
    \[a_*(a) = \max\{0\ ; a-X_{\textnormal{DF}}(a)\} \leq a - X(a) = a- \sum_{j=1}^N x_j(\tau\ ; a) \ .\]
    Since all $x_i(t\ ;a)$ are strictly monotonically increasing functions of time $t$, according to Theorem \ref{theorem: positive derivatives}, it follows that $a-\sum_{j=1}^N x_j(t\ ;a)$ is strictly monotonically decreasing. Hence, 
    \[a - \sum_{j=1}^N x_j(\tau\ ; a) \leq   a - \sum_{j=1}^N x_j(t\ ; a)\qquad \forall \ t\in [0,\tau]\]
    such that
    \[a_*(a)\tau = \int_0^\tau a_*(a) \ dt \leq \int_0^\tau \left( a-\sum_{j=1}^{N} x_j(t\ ;a)\right) dt\ ,\]
    which implies
    \begin{align*}
     X_{\textnormal{DF}}(a_*(a)) &= \sum_{i=1}^{N} g_i(1-e^{-k_{a,i}a_*(a)\tau})\\[1em]
    	& \leq \sum_{i=1}^N g_i(1-e^{-k_{a,i}\int_0^\tau \left( a -\sum_{j=1}^N x_j(t\ ;a)\right)  dt}) = X(a)\ . 
    \end{align*}
\end{proof}

Unfortunately, the bounds of Lemma \ref{lemma: weak bounds multi-epitope accumulation model} are rather weak and can become useless for some system parameters. For example, as discussed in subsection \ref{subsec: heuristics for antibody depletion}, the depletion effect can become very pronounced when $k_{a,i}\cdot\tau \cdot g_i$ is sufficiently large. In these cases, the depletion-free dose-response curve $X_{\textnormal{DF}}(a)$ is a weak upper bound that is much larger than the depletion dose-response curve $X(a)$. At least for a certain range of initial antibody concentrations. On the other hand, $a - X_{\textnormal{DF}}(a) > 0$ is only guaranteed for $a \geq \sum_{j=1}^N g_j$. Thus, $\max\{0,X_{\textnormal{DF}}(a_*(a))\}$ can be zero for a large range of initial antibody concentrations, which constitutes a weak lower bound. Figure \ref{fig: bounds 1} illustrates these problems. Note that the gray area, in which the depletion dose-response curve must be contained, does not help to gauge the shape of the depletion dose-response curve.

\begin{figure}[h!]
	\centering
	\begin{subfigure}[c]{0.49\textwidth}
		\caption{Linear scale}
		\includegraphics[width= \textwidth]{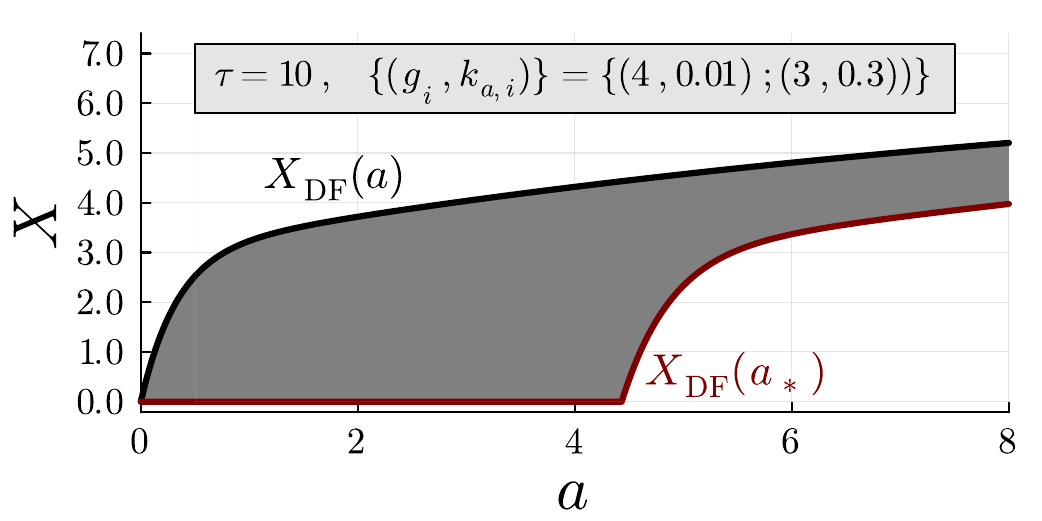}
	\end{subfigure}
	\begin{subfigure}[c]{0.49\textwidth}
		\caption{Logarithmic scale}
		\includegraphics[width= \textwidth]{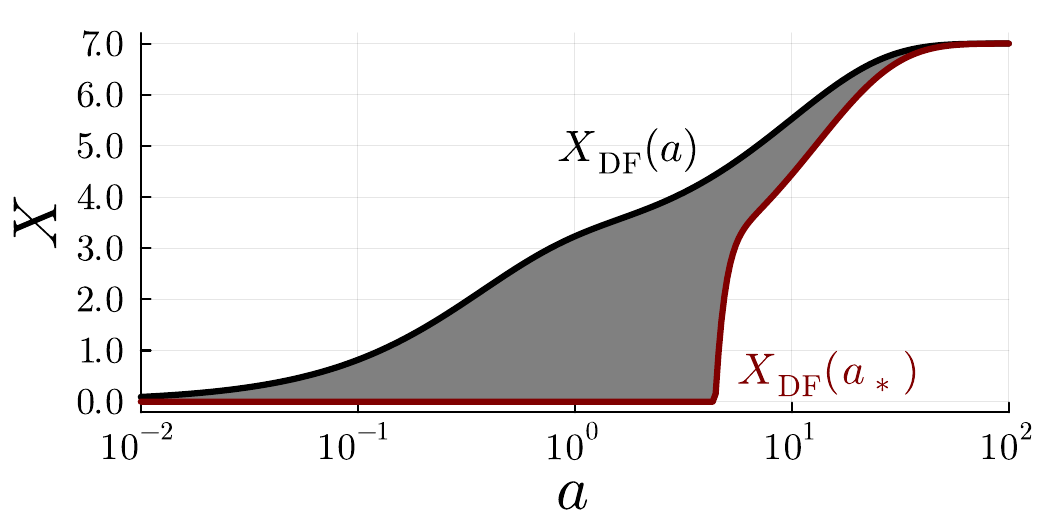}
	\end{subfigure}

	\caption{Plots for the upper and lower bounds from Lemma \ref{lemma: weak bounds multi-epitope accumulation model}. The gray area marks the region in which the depletion dose-response curve must be contained.} 
	\label{fig: bounds 1}
\end{figure}

\begin{figure}[h!]
	\centering
	\begin{tikzpicture}
		\filldraw[fill = gray!10, draw = black!50] (-2,-1) rectangle (10.5,5.5);

		\draw[line width = 1.5pt,-stealth] (0,0) -- (10,0) node[below]{$t$};
		\draw[line width = 1.5pt,-stealth] (0,0) -- (0,5) node[right]{$a$};

		\draw[line width = 1pt, dashed, color = black!60] (-0.3,4) node[left]{$a$} -- (9.8,4);
		\draw[line width = 1pt, dashed, color = black!60] (-0.3,4/5) node[left]{$\alpha(\tau\ ; a)$} -- (9.8,4/5);
		\draw[line width = 1pt, color = black!60] (8,-0.3) node[below]{$\tau$} -- (8,4.8);

		\node[color = red!50!black, above] at (9,0) {$\alpha(t\ ; a)$};

		\draw[line width = 2pt, color = blue!70!black] (0,4) -- (2,0) node[above right]{$a + t \frac{d \alpha}{dt}(0\ ; a)$};
		\draw[line width = 2pt,color = orange] (0,4) -- node[pos = 0.5, above right]{$a + t \frac{\alpha(\tau\ ; a)-a}{\tau}$} (8,4/5);

		\draw[line width = 2pt, color = red!50!black, domain = 0:9.8, samples = 100] plot({\x},{4/(0.5*\x + 1)});

	\end{tikzpicture}
	\caption{Illustration of the tangent line as lower bound of the antibody concentration and the secant line as upper bound for the antibody concentration.}
	\label{fig: convex antibody concentration}
\end{figure}

To obtain stronger bounds for the depletion accumulation model, we need to include more results from the last subsection and consider the time-development of the depletion antibody concentration more closely. Since all $x_i(t\ ;a)$ are concave and monotonically increasing, according to Theorem \ref{theorem: positive derivatives} and Corollary \ref{cor: concave accumulation}, the depletion antibody concentration $\alpha(t\ ; a) = a - \sum_{j=1}^N x_j(t\ ; a)$ is convex and monotonically decreasing. So, we can define a lower bound for $A(t,a)$, using the tangent line at $\alpha(0\ ; a)$ as a lower bound for the time development of the depletion antibody concentration. Figure \ref{fig: convex antibody concentration} illustrates the idea. In fact, since $\alpha(t\ ; a)$ is convex, the secant line between $(0,a)$ and $(\tau,\alpha(\tau\ ; a))$ would be an upper bound. However, since $\boldsymbol{x}(\tau\ ; a)$ is unknown, so is $\alpha(\tau\ ; a)$. Yet, to improve the upper bound, we recall theorem \ref{theorem: constant bounds for multi-epitope solution} which implies that $X(a)\leq a$. The following theorem summarizes the new bounds.

\begin{theorem}
	\label{theorem: improved bounds}
	Let $\{(g_i,k_{a,i})\}_{i=1}^N$ be epitope classes and $a>0$. Define
	\[\widetilde{A}(\tau\ ; a)\coloneqq \left\{ \begin{array}{ll}
		a \tau \left(1-\frac{\tau}{2}\sum_{i=1}^N k_{a,i}g_i\right) &\quad , \  \tau \leq \frac{1}{\sum_{i=1}^N k_{a,i}g_i}\\[1em]
		\frac{a}{2\sum_{i=1}^{N} k_{a,i}g_i} & \quad , \text{otherwise}
	\end{array} \right.\]
	and 
	\[Y(a) \coloneqq Y(\tau\ ; a) \coloneqq \sum_{i=1}^{N} g_i(1-e^{-k_{a,i}\widetilde{A}(\tau\ ; a)})\ .\]
	Then it holds that
	\[\max\{Y(a), X_{\textnormal{DF}}(a_*(a))\} \leq X(a)\leq \min\{a, X_{\textnormal{Df}}(a)\}\qquad \forall a > 0\ . \]
\end{theorem}
\begin{proof}
    According to Corollary \ref{cor: concave accumulation}, the $x_i(t\ ;a)$ are strictly convex as function of $t$. Hence, $X(t\ ; a) = \sum_{i=1}^N x_i(t\ ; a)$ is  strictly convex and $\alpha(t\ ; a) = a - X(t\ ; a)$ is strictly concave as function of $4$. According to Theorem \ref{theorem: positive derivatives}, all $x_i(t\ ; a)$ are strictly monotonically increasing, so $\alpha(t\ ; a)$ is strictly monotonically decreasing. This implies that $\frac{d}{dt} \alpha(t\ ; a) < 0$. From the convexity it thus follows that
    \[\frac{d a}{dt}(0\ ; a)\leq \frac{d}{dt} \alpha(t\ ; a)\qquad \forall\ t\in [0,\infty)\ .\]
    Using Lemma \ref{lemma: larger derivative and larger function}, we obtain
    \[\alpha(t\ ; a) - \left(a + t \frac{d a}{dt}(0\ ; a)\right) \geq 0 \qquad \forall\ t\in [0,\infty)\ ,\]
    which is equivalent to
    \[\alpha(t\ ; a)\geq \left(a + t \frac{d a}{dt}(0\ ; a)\right) \qquad \forall\ t\in [0,\infty)\ .\]
    From Theorem \ref{theorem: constant bounds for multi-epitope solution} it follows that $\alpha(t\ ; a) = a - \sum_{j=1}^N x_j(t\ ; a) \geq 0$ for all $t\in[0,\infty)$, implying
    \[\widetilde{\alpha}(t\ ; a) \coloneqq \max\left\{ 0\ ; a + t \frac{d a}{dt}(0\ ; a)\right\} \leq  \alpha(t\ ; a) \qquad \forall\  t \in [0,\infty)\ .\]
    
    Next, let us calculate at which point $\widetilde{\alpha}(t\ ; a) = 0$. For this, we use that $\boldsymbol{x}$ is a solution of the depletion accumulation IVP:
    \begin{align*}
    \frac{da}{dt}(0\ ; a) &= \left. \frac{d}{dt} \left(a -\sum_{i=1}^N x_i(t\ ; a)\right) \right|_{t = 0}\\[1em] 
    &= - \sum_{i=1}^N k_{a,i}\left(a -\sum_{j=1}^N x_j(0\ ; a)\right)(g_i - x_i(0\ ; a))\\[1em]
    &= -\sum_{i=1}^{N} k_{a,i} a g_i\ .
    \end{align*}
    In summary, we have
    \[\widetilde{\alpha}(t\ ; a) = \max\left\{0,a - a t \sum_{i=1}^N k_{a,i}g_i \right\} = 0 \qquad \Leftrightarrow \qquad t \geq \frac{1}{\sum_{i=1}^N k_{a,i}g_i} \eqqcolon t_0\ . \]
    Now, we can integrate $\widetilde{\alpha}(t\ ; a)$ over $[0,\tau]$, where $\tau\geq t_0$:
    \begin{align*}
    	\int_0^\tau \widetilde{\alpha}(t\ ; a) \ dt &= \int_0^{t_0} \left(a - a t \sum_{j=1}^N k_{a,i}g_i\right)  dt + \int_{t_0}^\tau 0 \ dt = \left[ a t - a \frac{t^2}{2}\sum_{i=1}^N k_{a,i}g_i  \right]_0^{t_0}\\[1em]
    	&= a t_0 \left(1- \frac{t_0}{2}\sum_{j=1}^N k_{a,i}g_i\right) = \frac{a}{\sum_{i=1}^N k_{a,i}g_i} \left(1-\frac{1}{2\sum_{i=1}^{N}k_{a,i}g_i} \sum_{i=1}^N k_{a,i}g_i\right) \\[1em]
    	&= \frac{a}{2\sum_{i=1}^N k_{a,i}g_i}\ .
    \end{align*}
    On the other hand, if $\tau < t_0$:
    \begin{align*}
    	\int_0^\tau \widetilde{\alpha}(t\ ; a) \ dt = \int_0^\tau \left( a - a t \sum_{j=1}^N k_{a,i}g_i\right)  dt = a \tau \left(1- \frac{\tau}{2}\sum_{j=1}^N k_{a,i}g_i\right)\ .
    \end{align*}
    Both integrals show that
    \[\int_{0}^{\tau} \widetilde{\alpha}(t\ ; a)\  dt = \widetilde{A}(\tau\ ; a)\ .\]
    Since we have already proven that $\widetilde{\alpha}(t\ ; a)\leq \alpha(t\ ; a)$ for all $t\geq 0$, it follows that
    \[\widetilde{A}(\tau\ ; a) = \int_0^\tau \widetilde{\alpha}(t\ ; a)\  dt \leq \int_0^\tau \alpha(t\ ; a) \ dt = A(\tau\ ; a)\quad \forall\ \tau \geq 0 \ .\]
    Thus, we have
    \[Y(a) = \sum_{i=1}^{N} g_i(1-e^{-k_{a,i}\widetilde{A}(\tau\ ; a)}) \leq \sum_{i=1}^{N} g_i(1-e^{-k_{a,i}A(\tau\ ; a)}) = X(a)\ .\]
    The remaining bounds follow from Theorem \ref{theorem: constant bounds for multi-epitope solution} and Lemma \ref{lemma: weak bounds multi-epitope accumulation model}. 
\end{proof}

We can apply the improved bounds to the same parameters that were used in Figure \ref{fig: bounds 1}. The results are shown in Figure \ref{fig: bounds 2}.

\begin{figure}[h!]
	\centering
	\begin{subfigure}[c]{0.49\textwidth}
		\caption{Linear scale}
		\includegraphics[width= \textwidth]{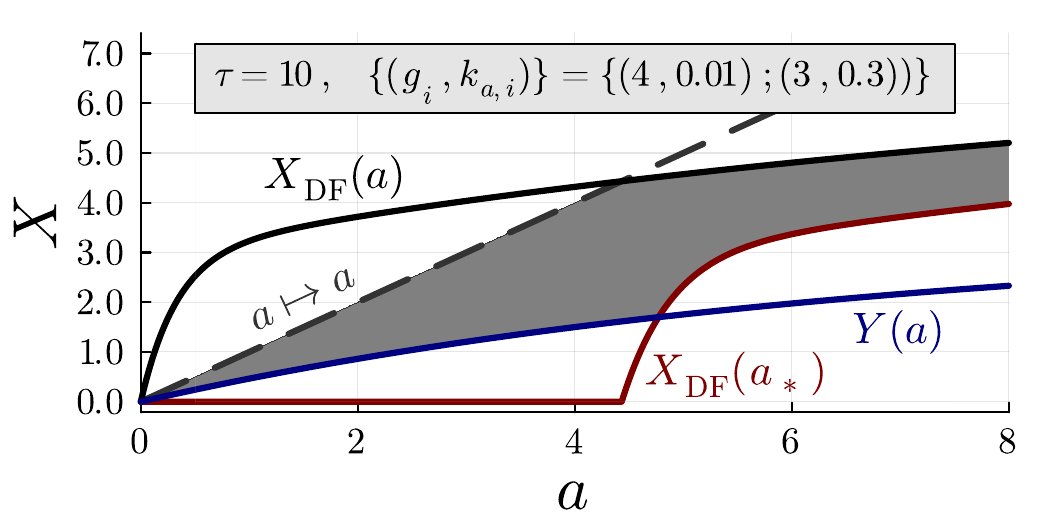}
	\end{subfigure}
	\begin{subfigure}[c]{0.49\textwidth}
		\caption{Logarithmic scale}
		\includegraphics[width= \textwidth]{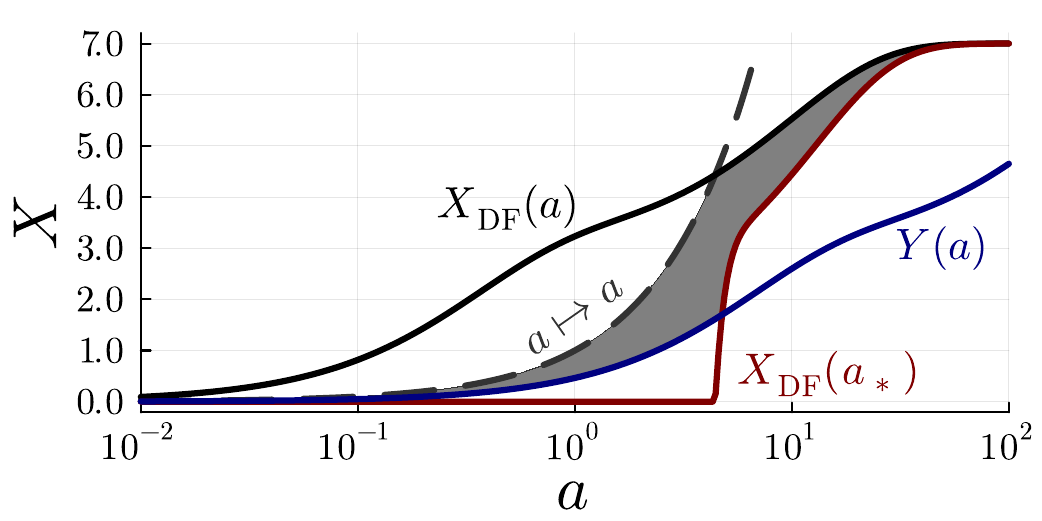}
	\end{subfigure}

	\caption{Plots for the improved upper and lower bounds from theorem \ref{theorem: improved bounds}, using the same system parameters as for Figure \ref{fig: bounds 1}. As before, the gray area marks the region in which the depletion dose-response curve must be contained.} 
	\label{fig: bounds 2}
\end{figure}

To investigate the bounds from Theorem \ref{theorem: improved bounds}, let us vary the system parameters. Figure \ref{fig: bounds 3} shows the plots for two cases: a ten-fold increase in the largest binding rate constant $k_{a,i}$ and a ten-fold decrease in both epitope concentrations.

\begin{figure}[h!]
	\centering
	\begin{subfigure}[c]{0.49\textwidth}
		\caption{Model plots linear scale}
		\includegraphics[width= \textwidth]{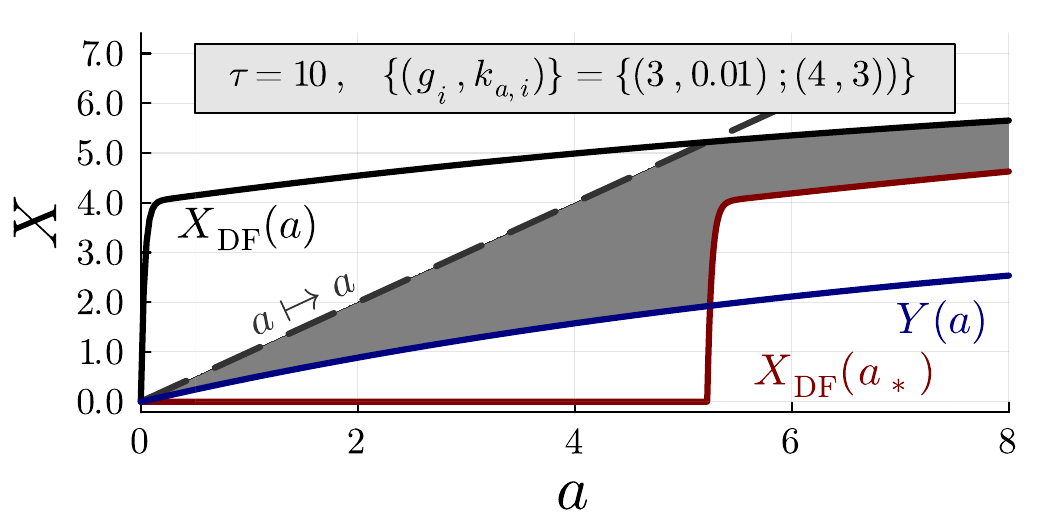}
	\end{subfigure}
	\begin{subfigure}[c]{0.49\textwidth}
		\caption{Model plots logarithmic scale}
		\includegraphics[width= \textwidth]{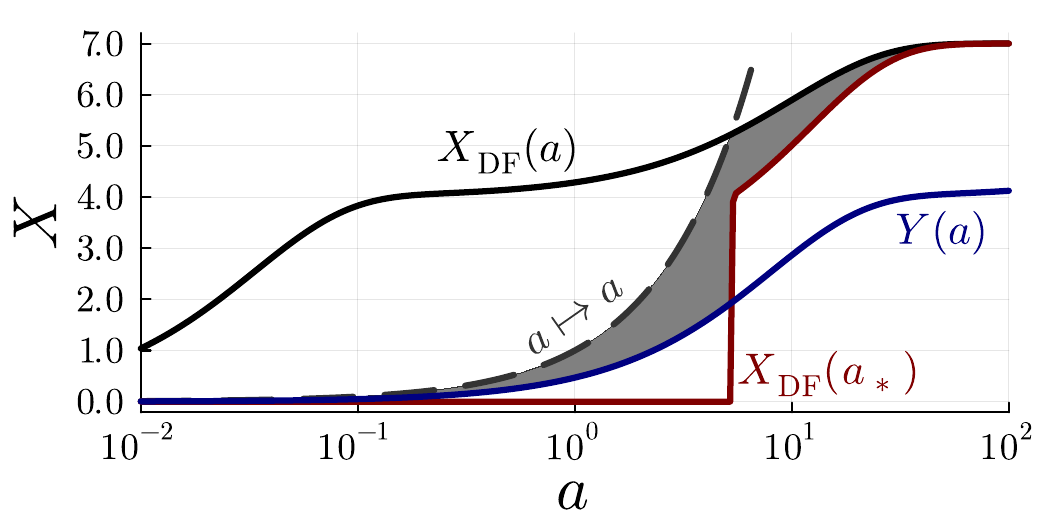}
	\end{subfigure}

	\begin{subfigure}[c]{0.49\textwidth}
		\caption{Model plots linear scale}
		\includegraphics[width= \textwidth]{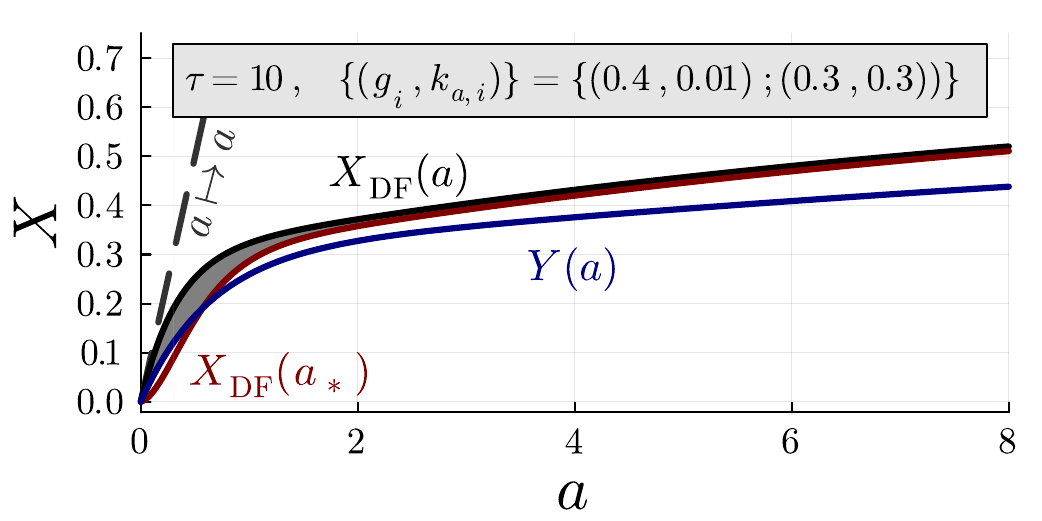}
	\end{subfigure}
	\begin{subfigure}[c]{0.49\textwidth}
		\caption{Model plots logarithmic scale}
		\includegraphics[width= \textwidth]{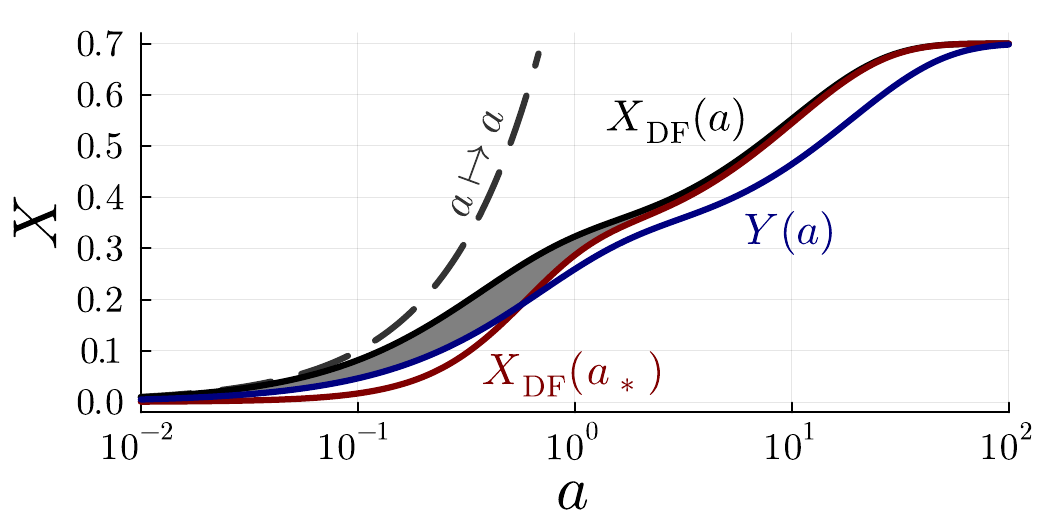}
	\end{subfigure}
	\caption{Plots for the improved bounds from theorem \ref{theorem: improved bounds} for different system parameters, compared to figure \ref{fig: bounds 2}. (a,b) show plots for a tenfold increase of the largest binding rate constant $k_{a,i}$. (c,d) show plots for a tenfold decrease of the epitope concentrations.} 
	\label{fig: bounds 3}
\end{figure}

Observe that increasing the binding rate constant has led to a stronger depletion effect. This can be explained with the initial antibody concentration, which is an upper bound of the depletion accumulation model but not of the depletion-free accumulation model. The latter has essentially an infinite amount of antibodies at its disposal. Thus, the concentration of bound antibodies increases, when the binding rate constants are increased. Since the initial antibody concentration is independent of the binding rate constants, the depletion accumulation model cannot follow, increasing the depletion effect.

Decreasing the epitope concentration, on the other hand, reduced the depletion effect. This is in line with the heuristics from subsection \ref{subsec: heuristics for antibody depletion}, which can be adapted to $\sum_{i=1}^{N} k_{a,i} g_i \tau$ as a rule of thumb for the depletion effect.

In addition, also the lower bounds improved for low epitope concentrations. To explain this observation, note that the approximation $a_*(a) = \max\{0,a -X_{\textnormal{DF}}(a)\}\approx a$  is valid for $a \gg X_{\textnormal{DF}}(a)$. This condition is essentially determined by $\sum_{i=1}^{N}g_i$, which is an upper bound of $X_{\textnormal{DF}}(a)$. Thus, $X_{\textnormal{DF}}(a_*(a)) \approx X_{\textnormal{DF}}(a)$, i.e. the difference between $X_{\textnormal{DF}}(a_*(a))$ and $X_{\textnormal{DF}}(a)$ decreases, when $\sum_{i=1}^{N}g_i$ decreases. Since $X_{\textnormal{DF}}(a_*(a)) \leq X(a)\leq X_{\textnormal{DF}}(a)$, this means that the lower bound improves.

Finally, note that we did not specify the number of epitope classes $N$ to derive the results of this section. Thus, the bounds of Theorem \ref{theorem: improved bounds} should also apply when there is only a single epitope class, where we have the analytical solution of the depletion accumulation model. This allows us to check the bounds by plotting them together with the true solution of the depletion accumulation model (Fig. \ref{fig: bounds single-epitope-class}).

\begin{figure}[h!]
	\centering
	\begin{subfigure}[c]{0.49\textwidth}
		\caption{Linear scale}
		\includegraphics[width= \textwidth]{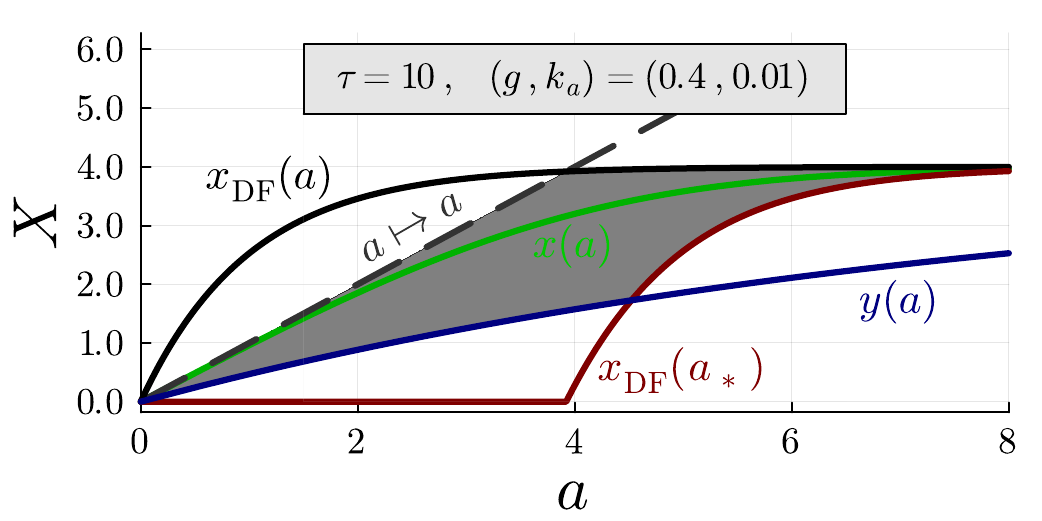}
	\end{subfigure}
	\begin{subfigure}[c]{0.49\textwidth}
		\caption{Logarithmic scale}
		\includegraphics[width= \textwidth]{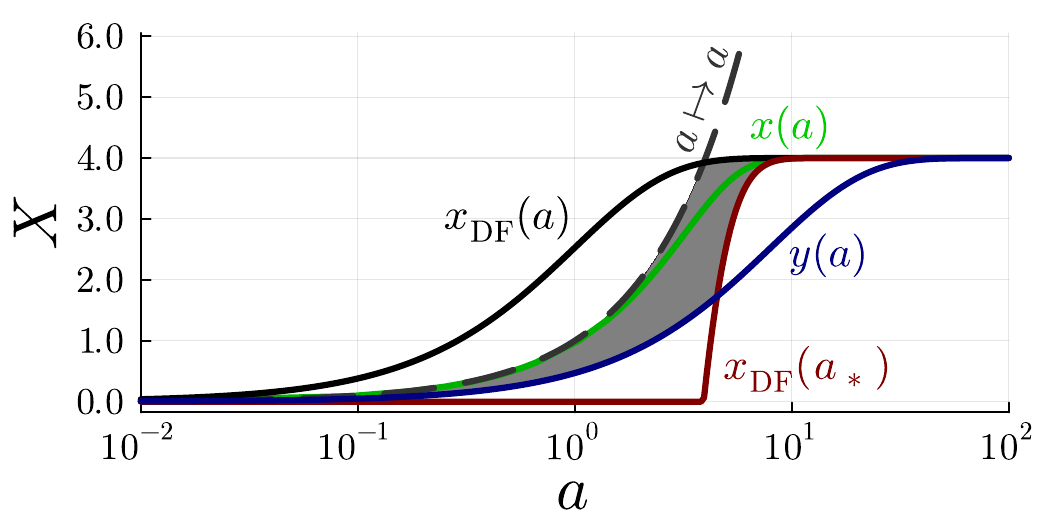}
	\end{subfigure}

	\caption{Plots of the improved bounds form theorem \ref{theorem: improved bounds} for a single epitope class, together with the analytical solution $x(a)$. Here, $y(a)$ denotes $Y(a)$ for the single-epitope-class case.} 
	\label{fig: bounds single-epitope-class}
\end{figure}

\subsection{Implications for inverse problems}

In inverse problems, the goal is to estimate the epitope classes $\{(g_i,k_{a,i})\}_{i=1}^N$ from dose-response data $\{(\mathfrak{a}_{i}, \mathfrak{X}_i)\}_{i=1}^m$. Here, $\mathfrak{a}_{i}$ denotes the initial antibody concentration of the $i$-th data point and $\mathfrak{X}_i$ denotes the correspondingly measured concentration of bound antibodies.
Since the epitope classes are unknown, the bounds from Theorem \ref{theorem: improved bounds} cannot be applied to inverse problems. However, we can do the opposite: Instead of approximating the depletion accumulation model, we can use the measured values $\mathfrak{X}_i$ to approximate the data that we would have obtained if there was no depletion. Then, after correcting the data points accordingly, the depletion-free accumulation model could be used for the inference problem. 

Without information about the epitope classes $\{(g_i,k_{a,i})\}_{i=1}^N$, it will be difficult to obtain approximations for the response values that would have arisen without depletion. However, using the measured response values, we can obtain bounds for the initial antibody concentrations that would have been necessary to obtain the measured response values in a system without depletion. Let us call these initial antibody concentrations \textbf{transformed initial antibody concentrations}. Using the convexity of the antibody concentration and the ideas depicted in the corresponding sketch (Figure \ref{fig: convex antibody concentration}) we can derive the following bounds.

\begin{theorem}
	Let $\{(\mathfrak{a}_{j}, \mathfrak{X}_j)\}_{j=1}^m$ be data points with $\mathfrak{a}_{j} > 0$ for all $j\in \{1,\ldots, m\}$, obtained from the depletion accumulation model for fixed but unknown epitope classes $\{(g_i,k_{a,i})\}_{i=1}^N$:
	\[\mathfrak{X}_j  = X(\tau\ ; \mathfrak{a}_{j}, \{(g_i,k_{a,i})\}_{i=1}^N)  \qquad \forall\ j \in \{1,\ldots,m\} \ .\]
	Let $\{\mathfrak{b}_{j}\}_{j=1}^m$ denote transformed initial antibody concentrations such that
	\[X_{\textnormal{DF}}(\tau\ ;\mathfrak{b}_{j}, \{(g_i,k_{a,i})\}_{i=1}^N) = X(\tau\ ; \mathfrak{a}_{j}, \{(g_i,k_{a,i})\}_{i=1}^N) \qquad \forall\ j  \in \{1,\ldots,m\} \ . \]
	Then it holds that
	\[\mathfrak{a}_{j} - \mathfrak{X}_j\  \leq\  \mathfrak{b}_{j}\ \leq\ \mathfrak{a}_{j} - \frac{1}{2}\mathfrak{X}_j \qquad \forall \ j \in  \{1,\ldots,m\}\ . \]
\end{theorem}

\begin{proof}
    To keep the notation short, we may drop the epitope classes $\{(g_i,k_{a,i})\}_{i=1}^N$ as extra arguments, as they are the same for all models being considered here.
    
    Theorem \ref{theorem: positive derivatives} and Corollary \ref{cor: concave accumulation} imply that the depletion antibody concentration $\alpha(t\ ; \mathfrak{a}_{j})$ is a strictly convex and monotonically decreasing function of $t$. It follows that
    \[\alpha(\tau\ ; \mathfrak{a}_{j}) = \mathfrak{a}_{j} - X(\mathfrak{a}_{j}) = \mathfrak{a}_{j}- \mathfrak{X}_j \leq \alpha(t\ ;\mathfrak{a}_{j})\]
    for all $t\in [0,\tau]$ and all $j \in \{1,\ldots, m\}$. Thus, we have
    \[(\mathfrak{a}_{j}-\mathfrak{X}_j)\tau = \int_0^\tau \mathfrak{a}_{j}-\mathfrak{X}_j \ dt \leq A(t\ ; \mathfrak{a}_{j})\]
    for all $t\in [0,\tau]$ and all $j \in \{1,\ldots, m\}$, implying that
    \begin{align*}
    	X_{\textnormal{DF}}(\mathfrak{a}_{j} - \mathfrak{X}_j) & = \sum_{i=1}^N g_i(1-e^{-k_{a,i} (\mathfrak{a}_{j}-\mathfrak{X}_j)\tau})\\[1em]
    	& \leq \sum_{i=1}^N g_{i}(1-e^{-k_{a,i}A(\tau\ ; \mathfrak{a}_{j})}) = X(\mathfrak{a}_{j})  \ \overset{\textnormal{def}}{=} \  X_{\textnormal{DF}}(\mathfrak{b}_{j})\ .
    \end{align*}
    Because of Theorem \ref{theorem: larger antibody concentrations} the $x_i(\tau\ ; a)$ and thus $X_{\textnormal{DF}}(\tau\ ; a)$ are monotonically increasing functions of $a$. This shows that $\mathfrak{a}_{j}-\mathfrak{X}_j < \mathfrak{b}_{j}$. So, $\mathfrak{a}_{j}- \mathfrak{X}_j$ are indeed lower bounds.
    
    For the upper bounds, we use again that $\alpha(t\ ; \alpha)$ is a convex function of $t$. Using the definition of convexity, we find 
    \[\mathfrak{a}_{j} - t \ \frac{\mathfrak{a}_{j}- \alpha(\tau\ ; \alpha)}{\tau} \geq \alpha(t\ ; \alpha)\qquad \forall \  t \in [0,\tau]\ , \ \forall\ j \in \{1,\ldots,m\}\ . \]
    With $\alpha(\tau\ ; \mathfrak{a}_{j}) = \mathfrak{a}_{j} - X(\mathfrak{a}_{j}) = \mathfrak{a}_{j} - \mathfrak{X}_j$ the last inequality implies
    \[\alpha(t\ ; \mathfrak{a}_{j})\leq \mathfrak{a}_{j} - t \ \frac{\mathfrak{a}_{j}, \alpha(\tau\ ; \alpha)}{\tau} = \mathfrak{a}_{j} - \frac{t}{\tau} \mathfrak{X}_j \ \]
    for all $t\in [0,\tau]$ and all $j\in \{1,\ldots, m\}$. Hence,
    \[A(\tau\ ; \mathfrak{a}_{j}) = \int_0^\tau \alpha(\tau\ ; \mathfrak{a}_{j})\ dt  \leq \int_0^\tau \mathfrak{a}_{j} - \frac{t}{\tau} \mathfrak{X}_j\ dt = \left(\mathfrak{a}_{j}-\frac{\mathfrak{X}_j}{2}\right) \tau \]
    for all $j\in \{1,\ldots, m\}$. It follows for all $j\in \{1,\ldots,m\}$ that
    \begin{align*}
    X_{\textnormal{DF}}(\mathfrak{b}_{j})\ \overset{\text{def}}{=}\  X(\mathfrak{a}_{j}) & = \sum_{i=1}^N g_i(1-e^{-k_{a,i}A(\tau\ ; \mathfrak{a}_{j})})\\[1em] 
    & \leq \sum_{i=1}^N g_i(1-e^{-k_{a,i}(\mathfrak{a}_{j}-\frac{\mathfrak{X}_j}{2})\tau}) = X_{\textnormal{DF}}(\mathfrak{a}_{j}-\tfrac{1}{2}\mathfrak{X}_j)\ .
    \end{align*}
    This proves that $\mathfrak{b}_j \leq \mathfrak{a}_{j} - \frac{1}{2}\mathfrak{X}_j$. So $\mathfrak{a}_{j} - \frac{1}{2}\mathfrak{X}_j$ is indeed an upper bound. 
\end{proof}

The lower bound $\mathfrak{a}_{j} - \mathfrak{X}_j$  was already used in (\PreprintText) but without formal proof.

So far, we assumed that the units of all quantities are converted to surface concentrations (see Remark \ref{rem: unit conversion}). If this is not the case, the conversion factor $\eta$ between the unit of $\mathfrak{a}_{j}$ and the unit of $\mathfrak{X}_j$ is required: $\mathfrak{a}_{j} - \eta \mathfrak{X}_j$. If there is substantial noise, or if the conversion factor is unknown, as was the case in (\PreprintText), the conversion factor must be estimated. In this case, lower bounds for the transformed initial antibody concentrations can be obtained by choosing the largest $\eta$ such that all $\mathfrak{a}_{j}-\eta \mathfrak{X}_j$ are positive.

Similarly, if different units are used, the upper bound becomes $\mathfrak{a}_{j}- \eta\frac{1}{2}\mathfrak{X}_j$. However, the upper bound cannot be used if the unit conversion factor $\eta$ is unknown. Here, instead of the largest possible $\eta$, one would have to use the smallest possible unit conversion factor to preserve the inequality for the upper bound. Assuming sensible units, the unit conversion factor must be a positive number. But there is no smallest positive unit conversion factor and the infimum is $\eta = 0$, i.e. no transformation at all.

\section{Conclusion and outlook}

In this paper, the depletion accumulation IVP was solved analytically, which yielded an analytical expression for the depletion accumulation model. With this expression, we derived inequalities between the depletion-free accumulation model, the depletion accumulation model, and the Langmuir isotherm. Although there seems to be no general analytical solution for the depletion accumulation IVP, existence, uniqueness and additional properties of the solution were proven. Among others, these properties were expected from the analytical solution for a single epitope class. We used the properties to derive bounds for the depletion accumulation model, which allowed us to investigate the depletion dose-response curve qualitatively. Finally, inverse problems were discussed and bounds for data transformations that undo the depletion effect were derived. 

Although the depletion behavior was characterized in this paper, exact statements with formal proofs are still missing. At least in case of a single epitope class, the analytical solution suggests that the described heuristic may be proven. However, for the case of multiple epitope classes, it is not clear if accordingly generalized statements could be proven.

Finally, numerical solutions and properties of inverse problems need to be considered in greater detail. Further research of properties and limitations is necessary. In particular, the use of numerical solutions for the depletion accumulation model should be considered.

%% file: appendix.tex
\section{Additional steps for some proofs}

\subsection{Proof of Lemma \ref{lemma: depletion accumulation IVP single epitope class}}
\label{appendix-proof: IVP solutions}

\subsubsection*{Generic case $a\neq g$}
When $a<g$, it follows that $e^{k_a (a-g)t} \leq 1$ and thus $g-a e^{k_a (a-g)t} > 0$. On the other hand, when $a>g$, it follows that $e^{k_a (a-g)t} \geq 1$ and thus $g-a e^{k_a (a-g)t} > 0$. So, $x(t; a)$ is finite for all $t\geq 0$.

First, we calculate the derivative of $x(t\ ; a)$. Let us define $E(t) = e^{k_a(a-g)t}$ to shorten the notation. We observe that
\[\frac{d}{dt}E(t) = \frac{d}{dt}e^{k_a(a-g)t} = k_a(a-g)e^{k_a(a-g)t} = k_a(a-g)E(t)\ .\]
Thus, we calculate
\begin{align*}
	\frac{d}{dt}x(t\ ;a) &= \frac{d}{dt}\frac{a g(1-E(t))}{g-a E(t)} \\[1em]
	&= \left(\frac{d}{dt} a g(1-E(t))\right)(g-a E(t))^{-1} + ag(1-E(t))\frac{d}{dt}\left(g-aE(t)\right)^{-1}\\[1em]
	&= -\frac{k_a(a-g)agE(t)}{(g-aE(t))}+ ag(1-E(t)) \frac{-1}{(g-aE(t))^2}(-ak_a(a-g)E(t))\\[1em]
	&= \frac{-k_a(a-g)agE(t)(g-aE(t))+ a g(1-E(t))ak_a(a-g)E(t)}{(g-aE(t))^2}\\[1em]
	&= \frac{k_a a g(a-g)\left(-gE(t) + a(E(t))^2+aE(t)-a(E(t))^2\right)}{(g-aE(t))^2}\\[1em]
	&= \frac{k_a a g (a-g)^2E(t)}{(g-aE(t))^2} = \frac{k_a a g (a-g)^2 e^{k_a(a-g)t}}{(g-a e^{k_a(a-g)t})^2}\ .
\end{align*} 
From the observations of the beginning of the proof it also follows that $\frac{d}{dt} x(t\ ; a)$ is finite for all $t\geq 0$.

Finally, we can plug in the expression for $x(t\ ;a)$ into the differential equation, using again the term $E(t) = e^{k_a(a-g)t}$ to shorten the notation:
\begin{align*}
	(a-x(t\ ; a))&(g-x(t\ ; a)) = ag - (a+g)x(t\ ; a) + (x(t\ ; a))^2\\[1em]
	&= a g - \frac{ag(a+g)(1-E(t))}{g-a E(t)} + \frac{a^2g^2(1-E(t))^2}{(g-aE(t))^2}\\[1em]
	& = \frac{ag(g-aE(t))^2- a g (a+g)(1-E(t))(g-aE(t)) + a^2g^2(1-E(t))^2}{(g-aE(t))^2}\\[2em]
	&= \frac{ \begin{array}{c}
		\cancel{ag^3}	- 2a^2g^2E(t) + \cancel{a^3g(E(t))^2}-\cancel{a^2g^2} \\ 
		+ a^3gE(t) + \cancel{a^2g^2E(t)}-\cancel{a^3g(E(t))^2} - \cancel{ag^3} + \cancel{a^2g^2E(t)}\\ 
		 + ag^3E(t) - \cancel{a^2g^2(E(t))^2} +\cancel{a^2g^2}-\cancel{2a^2g^2E(t)} + \cancel{a^2g^2(E(t))^2}
	\end{array}}{(g-aE(t))^2}\\[2em]
	&= \frac{ag^3E(t)+a^3gE(t)-2a^2g^2E(t)}{(g-aE(t))^2}\\[1em]
	&= \frac{ag(a-g)^2E(t)}{(g-aE(t))^2}\\[1em]
    &= \frac{d}{dt} x(t\ ;a,g)\ .
\end{align*}
Thus, it only remains to check the initial value:
\[x(0; a) = \frac{a g - a g e^{k_a (a-g)0}}{g - a e^{k_a (a -g)0}} = \frac{a g - a g}{g-a} = 0 \ .\]

\subsubsection*{Special case $a = g$}
Again, $x(t\ ; g)$ is well-defined and finite for all $t\geq 0$. As before, we calculate the derivative of $x(t\ ;g)$:
\begin{align*}
	\frac{d}{dt}x(t\ ;g) &= \frac{d}{dt} \frac{g}{1+\frac{1}{k_a g t}} = \frac{d}{dt} \frac{k_a g^2 t}{1+k_a g t}\\[1em] 
	&= \frac{k_a g^2}{1 + k_a g t} - \frac{k_a g^2 t}{(1+k_a g t)^2} k_a g = \frac{k_a g^2(k_a g t+1) - k_a^2g^3 t}{(1+k_a g t)^2}\\[1em] 
	&= \frac{k_a g^2}{(1+k_a g t)^2}\ ,
\end{align*}
which is also well-defined and finite for all $t\geq 0$. Next, we plug in the expression for $x(t\ ;g)$ into the differential equation:
\begin{align*}
	k_a(g-x(t\ ; g))^2 &= k_a \left(g-\frac{g}{1+\frac{1}{k_a gt}}\right)^2 = k_a \left(g-\frac{k_ag^2 t}{k_a g t+1}\right)^2\\[1em] 
	& = k_a \left(\frac{k_a g^2t + g - k_ag^2t}{k_a g t + 1}\right)^2 = \frac{k_a g^2}{(1+k_a g t)^2}\\[1em] 
	&= \frac{d}{dt}x(t\ ;g)\ .
\end{align*}
It remains to check the initial value:
\[x(0\ ;g)\coloneqq \lim_{t\searrow 0} \frac{g}{1+\frac{1}{k_a g t}} = 0\ .\]
\qed

\subsection{Proof of Theorem \ref{thm: bounded solutions single epitope}}

Since $k_a, a, t \geq 0$ it follows that
\[0\leq x_{\textnormal{DF}}(t) = g(1-e^{-k_a a t}) \leq g\ .\]
For $x(t)$, and $a = g$ it also follows that 
\[0 \leq x(t\ ; g, (g,k_a)) = \frac{g}{1+\frac{1}{k_a g t}} \leq g\ .\]

Let now $a < g$. Then $e^{k_a(a-g)t} \leq 1$ and thus
\[g-a e^{k_a(a-g)t}\geq 0 \qquad \text{and}\qquad 1-e^{k_a(a-g)t}\geq 0\ .\]
Hence, we have $0\leq x(t\ ; a ,(g,k_a))$. Next, we calculate 
\[\frac{g}{x(t\ ; a , (g,k_a))} = \frac{g- ae^{k_a(a-g)t}}{a(1-e^{k_a (a-g)t})} \leq  \frac{a- ae^{k_a(a-g)t}}{a(1-e^{k_a (a-g)t})} = 1\ . \]
Thus, it holds that $ x(t\ ; a , (g,k_a))\leq g$.

If $a > g$, it follows that $e^{k_a(a-g)t} \geq 1$ and thus 
\[g-a e^{k_a(a-g)t}\leq 0 \qquad \text{and}\qquad 1-e^{k_a(a-g)t}\leq 0\ .\]
Again, this implies $0\leq x(t\ ; a , (g,k_a))$ since both the numerator and denominator are negative.
Finally, we calculate
\[\frac{g}{x(t\ ; a , (g,k_a))} = \frac{g- ae^{k_a(a-g)t}}{a(1-e^{k_a (a-g)t})} \leq  \frac{g- ge^{k_a(a-g)t}}{a(1-e^{k_a (a-g)t})} = \frac{g}{a} \leq 1\ , \]
which shows that $x(t\ ; a , (g,k_a))\leq g$.

Since $x(t\ ; a, (g,k_a)) = x(t\ ; g, (a,k_a))$ by Lemma \ref{lemma: exchanging a and g}, the same steps also show that $x(t\ ; a , (g,k_a))\leq a$, which completes the proof.
\qed

\subsection{Proof for Theorem \ref{thm: accumulation model single epitope class}}

The existence and the expressions for $x_{\textnormal{DF}}(\tau\ ; a)$ and $x(\tau\ ; a)$ follow from Theorem \ref{thm: unique solution single epitope class} and definition \ref{def: accumulation models}.

Let $\tau > 0$, then the continuous differentiability and the expression for the derivative $\frac{d}{da} x_{\textnormal{DF}}(\tau\ ; a)$ are obvious. 

Next, we calculate $\frac{d}{da}x(\tau\ ; a)$ at $a \neq g$, which we denote by $\left. \frac{d}{da}x(a)\right|_{a \neq g}$:
\begin{align*}
	\left.\frac{d}{da} x(a)\right|_{a\neq g}  &= \frac{d}{da} \frac{a g(1-e^{k_a(a-g)\tau})}{g-ae^{k_a(a-g)\tau}}\\[1em]
	& = \frac{g(1-e^{k_a (a-g)\tau}) - k_a a g \tau e^{k_a(a-g)\tau}}{g-a e^{k_a(a-g)\tau}}\\[1em]
	&\qquad - \frac{a g(1-e^{k_a(a-g)\tau})(-e^{k_a(a-g)\tau} - k_a a \tau e^{k_a(a-g)\tau})}{(g-ae^{k_a(a-g)\tau})^2}\\[2em]
	&= \frac{ \begin{array}{c}
		(g(1-e^{k_a (a-g)\tau}) - k_a a g \tau e^{k_a(a-g)\tau})(g-ae^{k_a(a-g)\tau})\\
	+ a g(1-e^{k_a(a-g)\tau})(e^{k_a(a-g)\tau} + k_a a \tau e^{k_a(a-g)\tau})
	\end{array}
	 }{(g-ae^{k_a(a-g)\tau})^2}\\[2em]
	&= \frac{\begin{array}{c}
		 g^2 - g^2e^{k_a(a-g)\tau} - k_a a g^2\tau e^{k_a(a-g)\tau} - \cancel{age^{k_a(a-g)\tau}} + \cancel{a g e^{2k_a(a-g)\tau}}  \\
	+ \cancel{k_a a^2 g \tau e^{2k_a(a-g)\tau}} + \cancel{age^{k_a(a-g)\tau}} + k_a a^2 g \tau e^{k_a(a-g)\tau}\\
	 - \cancel{a g e^{2k_a(a-g)\tau}} - \cancel{k_a a^2 g \tau e^{2k_a(a-g)\tau}}
	\end{array} }{(g-ae^{k_a(a-g)\tau})^2}\\[2em]
	&= \frac{g^2 - g^2e^{k_a(a-g)\tau} - k_a a g^2 \tau e^{k_a(a-g)\tau} + k_a a^2 g \tau e^{k_a(a-g)\tau}}{(g-ae^{k_a(a-g)\tau})^2}\\[1em]
	&= \frac{ g(g+e^{k_a (a-g)\tau}(-g-k_a a g \tau + k_a a^2 \tau))}{(g-ae^{k_a(a-g)\tau})^2}\\[1em]
	&= \frac{ g(g-e^{k_a (a-g)\tau}(+g+k_a a g \tau - k_a a^2 \tau))}{(g-ae^{k_a(a-g)\tau})^2}\\[1em]
	&= \frac{g(g-e^{k_a(a-g)\tau}(k_a a \tau (g-a) + g))}{(g-ae^{k_a(a-g)\tau})^2}\quad .
\end{align*}
The derivative of $x(\tau\ ; a)$ at $a=g$, which we denote by $\left. \frac{d}{da} x(a) \right|_{a = g}$, needs to be calculated from the definition of derivatives. Here, we use the informal notation $\frac{0}{0}$ to indicate the use of L'H\^{o}pital's rule. Furthermore, note that we strike through terms when they cancel out additively (not when they cancel out between numerator and denominator):

\begin{align*}
	\left.\frac{d}{d a } x(a) \right|_{a = g}\  &\overset{\text{def}}{=}\  \lim_{a \rightarrow g} \frac{x(\tau\ ; a)-x(\tau\ ; g)}{a-g} = \lim_{a \rightarrow g} \frac{\frac{a g\left(1-e^{k_a(a -g)\tau}\right)}{g-a e^{k_a(a-g)\tau}} - \frac{k_a g^2 \tau}{1+k_a g \tau}}{a -g}\\[1em]
	&= \lim_{a \rightarrow g} \frac{\begin{array}{c}
		a g + k_a a g^2 \tau - a g e^{k_a(a-g)\tau} - k_a ag^2\tau e^{k_a(a-g)\tau}\\
		-k_a g^3\tau + k_a a g^2 \tau e^{k_a(a-g)\tau}
	\end{array} }{(a -g)(g+k_a g^2 \tau - a e^{k_a(a-g)\tau} - k_a a g \tau e^{k_a(a-g)\tau})}\\[2em]
	&= \frac{g^2 + k_a g^3\tau - g^2 e^0 - k_a g^3\tau e^0 - k_a g^3 \tau + k_a g^3\tau e^0}{0 (g+k_ag^2\tau- ge^0 - k_a g^2\tau e^0)}\\[1em]
	&= \frac{0}{0}
\end{align*}
\begin{align*}
	\Rightarrow\quad& \left.\frac{d}{d a } x(a) \right|_{a = g}  = \ldots \\[1em]
	&= \lim_{a \rightarrow g} \frac{\begin{array}{c}
		g + k_a g^2 \tau - ge^{k_a(a-g)\tau} - k_a a g \tau e^{k_a (a-g)\tau} - \cancel{k_ag^2\tau e^{k_a(a-g)\tau}}\\
		- \cancel{k_a^2 a g^2 \tau^2 e^{k_a(a-g)\tau}} + \cancel{k_ag^2\tau e^{k_a(a-g)\tau}} + \cancel{k_a^2 a g^2 \tau^2 e^{k_a(a-g)\tau}}
	\end{array} }{\begin{array}{c}
		g + k_ag^2\tau - 2ae^{k_a(a-g)\tau} - k_a a^2 \tau e^{k_a(a-g)\tau} - \cancel{2}k_a a g \tau e^{k_a(a-g)\tau}\\
		- k_a^2 a^2 g \tau^2 e^{k_a(a-g)\tau} + g e^{k_a(a-g)\tau}\\
		 + \cancel{k_a a g \tau e^{k_a(a-g)\tau}} + k_a g^2 \tau e^{k_a(a-g)\tau} + k_a^2 a g^2 \tau^2 e^{k_a(a-g)\tau}
	\end{array} }\\[2em]
	&= \lim_{a \rightarrow g} \frac{g + k_a g^2 \tau - g e^{k_a(a-g)\tau} - k_a a g \tau e^{k_a(a-g)\tau}}{\begin{array}{c}
		g + k_a g^2 \tau -2a e^{k_a(a-g)\tau} - k_a a^2 \tau e^{k_a(a-g)\tau}\\ 
		 -k_a a g \tau e^{k_a(a-g)\tau} -k_a^2 a^2 g \tau^2 e^{k_a(a-g)\tau} + g e^{k_a(a-g)\tau}\\ 
		  + k_a g^2 \tau e^{k_a(a-g)\tau} + k_a^2 a g^2 \tau^2 e^{k_a(a-g)\tau}
	\end{array} }\\[2em]
	&= \frac{g + k_a g^2 \tau - g e^0 - k_a g^2 \tau e^0}{\begin{array}{c}
		g+ k_a g^2 \tau - 2g e^0 - k_a g^2 \tau e^0 - k_a g^2 \tau e^0\\ 
		- k_a g^3 \tau^2 e^0 + g e^0 + k_a g^2 \tau e^0 + k_a g^3 \tau^2 e^0
	\end{array}
		}\\[1em]
	&= \frac{0}{0}
\end{align*}
\begin{align*}
	\Rightarrow\quad & \left.  \frac{d}{d a } x(a) \right|_{a = g}  = \ldots \\[1em]
	&= \lim_{a \rightarrow g} \frac{\begin{array}{c}
		- k_a g \tau e^{k_a(a-g)\tau} - k_a g \tau e^{k_a(a-g)\tau} - k_a^2 a g \tau^2 e^{k_a(a-g)\tau}
	\end{array} }{\begin{array}{c}
		-2e^{k_a(a-g)\tau} - 2k_a a \tau e^{k_a(a-g)\tau} - 2k_a a \tau e^{k_a(a-g)\tau}\\
		 - k_a^2 a^2 \tau^2 e^{k_a(a-g)\tau} - \cancel{k_a g \tau e^{k_a(a-g)\tau}} - k_a^2 a g \tau^2 e^{k_a(a-g)\tau}\\
		 -2 k_a^2 a g \tau^2 e^{k_a(a-g)\tau} - k_a^3 a^2g \tau^3 e^{k_a(a-g)\tau} + \cancel{k_a g \tau e^{k_a(a-g)\tau}}\\
		  + k_a^2 g^2 \tau^2 e^{k_a(a-g)\tau} + k_a^2 g^2 \tau^2 e^{k_a(a-g)\tau} + k_a^3 a g^2 \tau^3 e^{k_a(a-g)\tau}
	\end{array} }\\[2em]
	&= \frac{-2k_a g \tau e^0 - k_a^2 g^2 \tau^2 e^0}{\begin{array}{c}
		-2e^0 - 2k_a g \tau e^0 - 2k_a g \tau e^0\\
		 - k_a^2 g^2 \tau^2 e^0 - k_a^2 g^2 \tau^2 e^0\\
		 -\cancel{2 k_a^2 g^2 \tau^2 e^0} - \cancel{k_a^3 g^3 \tau^3 e^0} \\
		  + \cancel{k_a^2 g^2 \tau^2 e^0 + k_a^2 g^2 \tau^2 e^0} + \cancel{k_a^3 g^3 \tau^3 e^0}
	\end{array} }\\[1em]
	&= \frac{-2k_a g\tau -k_a^2 g^2 \tau^2}{-2 -4 k_a g \tau - 2k_a^2 g^2 \tau^2}\\[1em]
	&= \frac{k_a g \tau (2+ k_a g \tau)}{2(1+k_a g \tau)^2}\quad .
\end{align*}

\noindent
Now, we need to check that $\frac{d}{d a}x(\tau\ ; a)$ is continuous in $a = g$. As before, we use the informal notation $\frac{0}{0}$ to indicate the use of L'H\^{o}pital's rule. Furthermore, note that we strike through terms when they cancel out additively (not when they cancel out between numerator and denominator):

\begin{align*}
	\lim_{a \rightarrow g} \left. \frac{d}{d a} x(a)\right|_{a \neq g} &= \lim_{a \rightarrow g} \frac{g(g-e^{k_a(a-g)\tau}(k_a a \tau (g-a) + g))}{(g-ae^{k_a(a-g)\tau})^2}\\[1em]
	&= \lim_{a \rightarrow g}\frac{g^2 - g^2e^{k_a(a-g)\tau} - k_a a g^2 \tau e^{k_a(a-g)\tau} + k_a a^2 g \tau e^{k_a(a-g)\tau}}{g^2 + a^2 e^{2k_a(a-g)\tau} - 2 a g e^{k_a(a-g)\tau}}\\[1em]
	&= \frac{g^2- g^2 e^0 - k_a g^3 \tau e^0 + k_a g^3 \tau e^0}{g^2 + g^2 e^0 - 2g^2 e^0} = \frac{0}{0}
\end{align*}

\begin{align*}
	\Rightarrow\quad \lim_{a \rightarrow g} \left. \frac{d}{d a} x(a)\right|_{a \neq g} &= \lim_{a\rightarrow g} \frac{\begin{array}{c}
		-k_a g^2 \tau e^{k_a(a-g)\tau} - k_a g^2 \tau e^{k_a(a-g)\tau} - k_a^2 a g^2 \tau^2 e^{k_a(a-g)\tau}\\
		+ 2k_a a g \tau e^{k_a(a-g)\tau} + k_a^2 a^2 g \tau^2 e^{k_a(a-g)\tau}
	\end{array} }{\begin{array}{c}
		2a e^{2k_a(a-g)\tau} + 2k_a a^2 \tau e^{2k_a(a-g)\tau}\\
		- 2 g e^{k_a(a-g)\tau} - 2k_a a g \tau e^{k_a(a-g)\tau}
	\end{array} }\\[2em]
	&= \frac{ -k_a g^2 \tau e^0 - k_a g^2 \tau e^0 - k_a^2g^3\tau^2 e^0 + 2k_a g^2\tau e^0 + k_a^2 g^3\tau^2 e^0 }{2ge^0 + 2 k_a g^2\tau - 2ge^0 - 2k_a g^2\tau e^0}\\[1em]
	&= \frac{0}{0}
\end{align*}

\begin{align*}
	\Rightarrow\ \lim_{a \rightarrow g} \left. \frac{d}{d a} x(a)\right|_{a \neq g} &= \lim_{a\rightarrow g} \frac{\begin{array}{c}
		-k_a^2g^2\tau^2 e^{k_a(a-g)\tau} -k_a^2g^2\tau^2 e^{k_a(a-g)\tau}  -k_a^2g^2\tau^2 e^{k_a(a-g)\tau}\\
		-k_a^3 a g^2\tau^3 e^{k_a(a-g)\tau} +2k_a g \tau e^{k_a(a-g)\tau} + 2k_a^2 a g \tau^2 e^{k_a(a-g)\tau}\\
		+ 2k_a^2a g \tau^2 e^{k_a(a-g)\tau} + k_a^3 a^2 g \tau^3 e^{k_a(a-g)\tau}
	\end{array} }{\begin{array}{c}
		2e^{2k_a(a-g)\tau} + 4k_a a \tau e^{2k_a(a-g)\tau} + 4k_a a \tau e^{2k_a(a-g)\tau}\\
		+ 4k_a^2 a^2 \tau^2 e^{2k_a(a-g)\tau} - 2k_a g \tau e^{k_a(a-g)\tau}\\
		- 2k_a g \tau e^{k_a(a-g)\tau} - 2k_a^2 a g \tau^2 e^{k_a(a-g)\tau}
	\end{array} }\\[2em]
	&= \frac{\begin{array}{c}
		-k_a^2g^2\tau^2 e^0 -k_a^2g^2\tau^2 e^0  -k_a^2g^2\tau^2 e^0\\
		-k_a^3 g^3\tau^3 e^0 +2k_a g \tau e^0 + 2k_a^2 g^2 \tau^2 e^0\\
		+ 2k_a^2 g^2 \tau^2 e^0 + k_a^3 g^3 \tau^3 e^0
	\end{array} }{\begin{array}{c}
		2e^0 + 4k_a g \tau e^0 + 4k_a g \tau e^0\\
		+ 4k_a^2 g^2 \tau^2 e^0 - 2k_a g \tau e^0\\
		- 2k_a g \tau e^0 - 2k_a^2 g^2 \tau^2 e^0
	\end{array} }\\[2em]
	& = \frac{-3k_a^2g^2\tau^2 - k_a^3g^3\tau^3 + 2k_a g\tau + 4 k_a^2g^2\tau^2 + k_a^3g^3\tau^3}{2+8 k_a g \tau + 4k_a^2 g^2 \tau^2 - 4k_a g \tau- 2k_a^2 g^2 \tau^2}\\[1em]
	&= \frac{2k_a g \tau + k_a^2 g^2 \tau^2}{2+ 4k_a g \tau + 2 k_a^2g^2\tau^2} = \frac{k_a g \tau(2+ k_a g \tau)}{2(1+k_a g \tau)^2} \\[1em]
	&= \left.\frac{d}{da} x(a)\right|_{a = g}\quad .
\end{align*}

To verify that the derivatives are finite, we recall that $g-ae^{k_a (a-g)\tau}$ is non-zero for all $k_a,a,g,\tau > 0$ when $a\neq g$, and that $1+k_a g \tau$ is non-zero for all $k_a,a,g,\tau > 0$. This was shown at the beginning of the proof of Lemma \ref{lemma: depletion accumulation IVP single epitope class} (appendix \ref{appendix-proof: IVP solutions}). Thus, all denominators that appear in the derivative $\frac{d}{da} x(\tau\ ; a)$  are non-zero for their respective cases, such that the derivative $\frac{d}{da} x(\tau\ ; a)$ finite and well-defined for all $k_a, a, g, \tau > 0$. Hence, as a rational function, it is also continuous in $a$ for $a>0$.

Finally, we consider the case $\tau = 0$. From the initial values of the accumulation IVPs it follows that $x_{\textnormal{DF}}(0\ ; a) = 0$ and that $x(0\ ; a) = 0$. Thus, the derivatives with respect to $a$ must also be zero. We complete the proof by checking that this holds true for the expressions of the derivatives that we have derived above. For this, we calculate
\[\lim\limits_{\tau \searrow 0} \frac{d}{da} x_{\textnormal{DF}}(\tau\ ; a) = \lim\limits_{\tau \searrow 0} k_a g \tau e^{-k_a a \tau} = 0\ , \]

\[\lim\limits_{\tau \searrow 0} \left. \frac{d}{da} x(\tau\ ; a) \right|_{a\neq g} = \lim\limits_{\tau \searrow 0} \frac{g(g-e^{k_a(a-g)\tau}(k_a a \tau (g-a) + g))}{(g-ae^{k_a(a-g)\tau})^2}  = 0 \ , \]
and
\[\lim\limits_{\tau \searrow 0} \left. \frac{d}{da} x(\tau\ ; a) \right|_{a = g} = \lim\limits_{\tau \searrow 0} \frac{k_a g \tau(2+ k_a g \tau)}{2(1+k_a g \tau)^2}  = 0 \ . \]

\qed

\subsection{Proof for lemma \ref{lemma: larger derivative and larger function}}
\label{proof: larger derivative and larger function}

For $t = a$ the statement is obvious. If $a=\max\{I\}$ the statement is proven. Otherwise, let $t>a$. Then the mean value theorem implies that there is a $c\in [a,t]$ such that
\[\frac{f(t)-f(a)}{t-a} = f'(c) .\]
Since $t>a$ and $f'(c)\geq 0$, by the assumptions of the lemma, it follows that 
\[f(t)-f(a)\geq 0 \qquad \Rightarrow \qquad f(t) \geq f(a)\quad \forall \  t > a\ .\]
For the strict inequality, repeat the argument with $f'(c) > 0$.
\qed